\documentclass{LMCS}

\def\doi{8 (1:08) 2012}
\lmcsheading%
{\doi}
{1--29}
{}
{}
{May\phantom.~12, 2011}
{Feb.~20, 2012}
{}

\usepackage{enumerate}

\usepackage{xspace}
\usepackage{amsmath}

\usepackage{xspace}
\usepackage{graphicx}

\usepackage{amssymb}
\usepackage{amsfonts}
\usepackage{proof}
\usepackage{epsfig}
\usepackage{verbatim}

\newcommand{\SB}[1]{[\![#1]\!]}

\newcommand{\ab}[1]{\langle #1 \rangle}

\newcommand{\impp}{\Longrightarrow}

\newcommand{\hyp}[2][]{\infer[#1]{#2}{}}

\newcommand{\new}[1][]{\reflectbox{\sf{#1{}N}}}
\newcommand{\smallnew}{\new[\scriptsize]}
\newcommand{\fresh}{\mathrel{\#}}

\newcommand{\abs}[2]{{\ab{#1}{#2}}}
\newcommand{\conc}{\mathop{@}}
\newcommand{\name}[1]{\mathsf{#1}}
\newcommand{\Aa}{\name{a}}
\newcommand{\Ab}{\name{b}}

\newcommand{\Ax}{\name{x}}
\newcommand{\Ay}{\name{y}}

\newcommand{\Az}{\name{z}}

\newcommand{\id}{\mathsf{id}}

\newcommand{\act}{\boldsymbol{\cdot}}

\newcommand{\labelSec}[1]{\label{sec:#1}}
\newcommand{\refSec}[1]{Section~\ref{sec:#1}}
\newcommand{\labelFig}[1]{\label{fig:#1}}
\newcommand{\refFig}[1]{Figure~\ref{fig:#1}}

\newcommand{\labelThm}[1]{\label{thm:#1}}
\newcommand{\refThm}[1]{Theorem~\ref{thm:#1}}
\newcommand{\labelLem}[1]{\label{lem:#1}}
\newcommand{\refLem}[1]{Lemma~\ref{lem:#1}}

\newcommand{\labelCor}[1]{\label{cor:#1}}
\newcommand{\refCor}[1]{Corollary~\ref{cor:#1}}

\renewcommand{\nd}{\vdash}

\newcommand{\DNTT}{\ensuremath{ \lambda^{ \Pi \smallnew}}\xspace}

\newcommand{\wfres}[4]{#1 \vdash \bind{#2}{#3} \mathrel{\backslash} #4}

\newcommand{\sigwf}[1]{\nd #1 ~ \mathsf{sig}}
\newcommand{\ctxwf}[1]{\nd #1~\mathsf{ctx}}

\newcommand{\kind}{\mathrm{kind}}
\newcommand{\kdwf}[2]{#1 \nd #2 : \kind}
\newcommand{\tywf}[3]{#1 \nd #2:#3}
\newcommand{\wf}[3]{#1 \nd #2: #3}

\newcommand{\wfsub}[3]{#1 \nd #2: #3}

\newcommand{\streq}[4]{\judge{#1}{#2\leftrightarrow#3:#4}}
\newcommand{\algeq}[4]{\judge{#1}{#2\Leftrightarrow#3:#4}}
\newcommand{\wfcan}[4]{\judge{#1}{#2 \Uparrow #3:#4}}
\newcommand{\wfatm}[4]{\judge{#1}{#2\downarrow #3:#4}}
\newcommand{\logrel}[3]{\judge{#1}{#2\in\SB{#3}}}
\newcommand{\judge}[2]{#1 \nd  #2 }
\newcommand{\JJ}{\mathcal{J}}
\newcommand{\DD}{\mathcal{D}}

\newcommand{\NEW}[3]{\new #1{:}#2.#3}
\newcommand{\PI}[3]{\Pi #1{:}#2.#3}
\newcommand{\LAM}[3]{\lambda #1{:}#2.#3}
\newcommand{\ABS}[3]{\abs{#1{:}#2}#3}
\newcommand{\Abs}[2]{\langle\! \langle #1 \rangle\! \rangle #2}
\newcommand{\bind}[2]{#1{:}#2}

\newcommand{\ato}[1]{\stackrel{#1}{\longrightarrow}}
\newcommand{\whr}{\ato{\mathrm{whr}}}
\renewcommand{\conc}{ @ }
\newcommand{\type}{\mathrm{type}}
\newcommand{\nametype}{\mathrm{name}}

\newcommand{\nat}{\mathit{nat}}
\newcommand{\option}{\mathit{option}}
\newcommand{\SOME}{\mathit{SOME}}
\newcommand{\NONE}{\mathit{NONE}}

\newcommand{\enc}[1]{\ulcorner#1\urcorner}

\newcommand{\rec}{\mathit{rec}}
\newcommand{\var}{\mathit{var}}
\newcommand{\lam}{\mathit{lam}}
\newcommand{\app}{\mathit{app}}

\newcommand{\varty}{\mathit{v}}
\newcommand{\expty}{\mathit{e}}

\newcommand{\neqq}{\mathit{neq}}
\newcommand{\neqqxx}[2]{\mathit{neq\_{#1}\_{#2}}}

\newcommand{\atomic}{\mathit{atomic}}

\newcommand{\tr}[1]{}

\begin{document}

\title[A dependent nominal type theory]{A dependent nominal type theory\rsuper*}

\author[J.~Cheney]{James Cheney}	
\address{Laboratory for Foundations of Computer Science, 
University of Edinburgh}	
\email{jcheney@inf.ed.ac.uk}  


\keywords{dependent types, nominal abstract syntax} 
\subjclass{F.4.1} 
\titlecomment{{\lsuper*}This article extends ``A simple
  nominal type theory'', published in LFMTP 2008
  ~\cite{cheney08lfmtp}}


\begin{abstract}
  \noindent
  Nominal abstract syntax is an approach to representing names and
binding pioneered by Gabbay and Pitts.  So far nominal techniques have
mostly been studied using classical logic or model theory, not type
theory.  Nominal extensions to simple, dependent and ML-like polymorphic 
languages have been studied, but decidability and normalization results
have only been established for simple nominal type theories.  We
present a LF-style dependent type theory extended with
name-abstraction types, prove soundness and decidability of
$\beta\eta$-equivalence checking, discuss adequacy and canonical forms
via an example, and discuss extensions such as dependently-typed
recursion and induction principles.

\end{abstract}

\maketitle

\section{Introduction}

Nominal abstract syntax, introduced by Gabbay and
Pitts~\cite{gabbay02fac,pitts03ic,pitts06jacm}, provides a relatively
concrete approach to abstract syntax with binding.  Nominal techniques
support built-in alpha-equivalence with the ability to compare names
as data, but (unlike higher-order abstract
syntax~\cite{harper93jacm,pfenning89pldi,nadathur98higher}) do not
provide built-in support for substitution or contexts. On the other
hand, definitions that involve comparing names as values are sometimes
easier to define using nominal abstract syntax, and both single and
simultaneous substitution can be defined easily as primitive
recursive functions over nominal abstract syntax (see
e.g.~\cite{cheney05icfp,cheney08lfmtp,pitts10popl}). Thus, nominal
abstract syntax is an alternative approach to representing languages
with bound names that has different strengths and weaknesses compared
to higher-order abstract syntax.

Historically, one weakness has been the absence of a clean
type-theoretic framework for nominal abstract syntax, paralleling
elegant frameworks such as LF~\cite{harper93jacm},
$\lambda$Prolog~\cite{nadathur98higher}, and more recently
Delphin~\cite{poswolsky08esop} and Beluga~\cite{pientka08popl}.  Some
previous steps have been taken towards nominal type theories
sufficient for reasoning about nominal abstract
syntax~\cite{schoepp04csl,cheney08lfmtp,pitts10popl,westbrook09lfmtp},
but as yet a full dependent type theory equipped with metatheoretic
results such as decidability of typechecking has not been developed.

In this article, we take a step towards such a nominal type theory, by
extending a previously-developed simply typed calculus~\cite{cheney08lfmtp} with
dependent types, roughly analogous to the LF system (though with some
different modes of use in mind).  We call our system \DNTT, or
\emph{dependent nominal type theory}.  \DNTT provides simple
techniques for encoding judgments that depend on name-distinctness and
can be soundly extended with recursion combinators
useful for defining functions and proofs involving nominal abstract
syntax.  Because \DNTT lacks
built-in support for substitution over nominal abstract syntax, it
should so far be viewed as a step towards dependently-typed
programming and reasoning with nominal features and not as a
self-contained logical framework like LF.  For example, our approach
could serve as a starting point (or domain-specific embedded language)
for dependently-typed programming with names and binding within
systems such as Agda or Coq based on constructive type theories, as
advocated by Licata et al.~\cite{licata08lics}, Westbrook et
al.~\cite{westbrook09lfmtp}, or Poulliard and
Pottier~\cite{poulliard10icfp}.

We add \emph{names} $\Aa,\Ab,\ldots$, \emph{name types} $\alpha$, and
a \emph{dependent name-abstraction} type constructor
$\NEW{\Aa}{\alpha}{B}$ to LF, which is introduced by
\emph{abstraction} ($\abs{\Aa}{M}$) and eliminated by
\emph{concretion} ($M \conc \Aa$). The abstraction term can be viewed
as constructing an $\alpha$-equivalence class that binds a name; the
concretion term instantiates the name bound by an abstraction to a
\emph{fresh} name $\Aa$.  This freshness requirement ensures that no
two (syntactically) distinct names can ever be identified via
renaming, so it is possible to reason about inequalities among names
in \DNTT.  Moreover, this restriction justifies a semantic
interpretation of name and name-abstraction types in \DNTT as names
and name-abstraction constructions in nominal logic, which in turn
justifies adding recursion combinators that can be used to define
functions on and reason about inductively-defined types with
name-binding within \DNTT.

\paragraph{Example}

As a simple example of a relation that is easily definable in \DNTT,
but cannot as easily be defined in LF, consider the signature in
\refFig{lambda} and alpha-inequivalence relation defined in
\refFig{alpha-ineq}.  (The notation $\Abs{\varty}{\expty}$ stands for
the non-dependent name-abstraction type $\NEW{\Aa}{\varty}{\expty}$.)
The key rules are $\neqqxx{v}{v}$ and $\neqqxx{l}{l}$; several other
symmetric rules are omitted.  Both rules use the $\new$-quantifier to
generate fresh names.  The type of $\neqqxx{v}{v}$ states that two
variables are alpha-inequivalent if their names are distinct.  The
type of $\neqqxx{l}{l}$ states that two lambda-abstractions are
alpha-inequivalent if their bodies are inequivalent when instantiated
to the same fresh name $\Aa$.  We discuss this example further in
\refSec{adequacy} and \refSec{discussion}.

\begin{figure}[tb]
\[\begin{array}{lclcl}
  \varty  &:& \nametype. \quad \expty : \type.
  \smallskip\\
  \var &:& \varty \to \expty.\quad   
  \app : \expty \to \expty \to \expty.\quad
  \lam : \Abs{\varty}{\expty} \to \expty.
\end{array}\]
  \caption{Lambda-calculus syntax in \DNTT}\labelFig{lambda}
 \[\small
  \begin{array}{rcl}
    \neqq&:& \expty\to \expty \to \type.\\
    \neqqxx{v}{v} &:& \new \Aa. \new \Ab. \neqq~(\var~\Aa)~(\var~\Ab).\\
    \neqqxx{a}{a}_i &:& \neqq~M_i~N_i \to \neqq~(\app~M_1~M_2)~(\app~N_1~N_2).\\
    \neqqxx{l}{l} &:& (\new \Aa. \neqq~(M\conc \Aa)~(N \conc \Aa)) \to \neqq~(\lam~M)~(\lam~N).\\
    \neqqxx{v}{a} &:& \neqq~(\var~V)~(\app~M~N).
    \\\neqqxx{v}{l} &:& \neqq~(\var~V)~(\lam~M).\\
    \neqqxx{a}{l} &:& \neqq~(\app~M~N)~(\lam~P). 
 \end{array}
\]
  \caption{Alpha-inequality in \DNTT}
  \label{fig:alpha-ineq}
\end{figure}

\paragraph{Contributions}
The main contribution of this article is the formulation of \DNTT and
the proof of key metatheoretic properties such as decidability of
typechecking, canonicalization, and conservativity over LF.  At a
technical level, our contribution draws upon Harper and Pfenning's
proof of these properties for LF~\cite{harper05tocl}, and we focus on
the aspects in which \DNTT differs from LF, primarily having to do
with the treatment of name-abstraction types and concretion via the
restriction judgment.

\paragraph{Outline} The structure of the rest of this article is as
follows.  \refSec{related} discusses additional related work.
\refSec{nlf} presents the \DNTT type theory, along with basic
syntactic properties.  \refSec{meta} develops the metatheory of \DNTT.
\refSec{adequacy} considers canonical forms and adequacy of
representations of nominal abstract syntax in \DNTT via a standard
example.  \refSec{examples} discusses several examples and extensions
such as recursion combinators.  \refSec{discussion} contrasts \DNTT
with closely related systems. \refSec{concl} discusses future work and
concludes.

\section{Related work}\labelSec{related}
Typed programming languages and type theories incorporating nominal
features have already been
studied~\cite{shinwell03icfp-short,schoepp04csl,schoepp06phd,pottier07lics,cheney08lfmtp,westbrook08phd}.
As in some previous
work~\cite{schoepp04csl,schoepp06phd,cheney08lfmtp,westbrook08phd}, we
employ bunched contexts~\cite{ohearn99bsl} to enforce the freshness
side-conditions on concretions.  Specifically,
following~\cite{cheney08lfmtp}, we employ an explicit context
restriction judgment in order to prevent references to the name $\Aa$
within $M$ in a concretion $M \conc \Aa$.
Previously~\cite{cheney08lfmtp}, we proved strong normalization for a
simple nominal type theory by translation to ordinary lambda-calculus.
Here, we prove completeness of a $\beta\eta$-equivalence algorithm
more directly by adapting Harper and Pfenning's logical-relations
proof for LF~\cite{harper05tocl}.  The restriction judgment is used
essentially in the modified logical relation.

Sch\"opp and Stark~\cite{schoepp04csl,schoepp06phd} and Westbrook et
al.~\cite{westbrook08phd,westbrook09lfmtp} have considered richer
nominal type theories than \DNTT.  However, Sch\"opp and Stark did not
investigate normalization or  decidability, whereas Westbrook proves
$\beta$-normalization for a ``Calculus of Nominal Inductive
Constructions'' (CNIC) by a (somewhat complex) translation to ordinary
CIC~\cite{westbrook08phd}; our logical-relations proof handles
$\beta\eta$-equivalence and seems more direct but does not deal with
inductive types or polymorphism.  Westbrook et al. are developing an
implementation of CNIC called Cinic~\cite{westbrook09lfmtp}.

Pitts~\cite{pitts10popl,pitts11jfp} has recently investigated a
``Nominal System T'' that extends simple nominal type
theory~\cite{cheney08lfmtp} with locally-scoped names ($\nu$-expressions)
and recursion over lambda-terms encoded using nominal abstract syntax.
Strong normalization modulo a structural congruence is proved via
normalization-by-evaluation.  An extended version of this
work~\cite{pitts11jfp} is different in some ways, and gives an
alternative proof of $\beta$-normalization.  Both techniques draw on
Odersky's $\lambda \nu$-calculus~\cite{odersky94popl}.

In Pitts' approach, contexts are standard and do not incorporate
freshness assertions, but as a result there are ``exotic'' terms such
as $\nu \Aa. \var~\Aa : \expty$, which do not correspond to any object
language term and complicate the argument for adequacy.  Nevertheless,
Pitts' approach is an interesting development that may lead to a more
expressive and flexible facilities for dependently-typed programming
with nominal abstract syntax.  However, as discussed in
\refSec{discussion}, there are potential complications in pushing this
approach beyond simple $\Pi$-types.

Our approach also bears some similarity to work on weak higher-order
abstract syntax, primarily employed in constructive type theories such
as Coq~\cite{despeyroux94lpar,despeyroux95tlca,schuermann01tcs}.  Here, in contrast to
ordinary higher-order abstract syntax the idea is to use a different,
atomic type for binders via a function space $\varty \to \expty$. The type
$\varty$ can be an abstract type with decidable equality; this
makes it possible to define the type of expressions inductively, but primitive recursion over weak HOAS is not
straightforward to incorporate into Coq.  This approach has been
formalized as a consistent extension called the Theory of
Contexts~\cite{honsell01tcs,DBLP:journals/jfp/BucaloHMSH06}, and this
theory has been related to nominal abstract syntax by Miculan et
al.~\cite{miculan05merlin}.  

There has also been recent work on techniques for recursion over
higher-order abstract syntax.  Pientka~\cite{pientka08popl}, Poswolsky
and Sch\"urmann~\cite{poswolsky08esop}, and Licata et
al.~\cite{licata08lics} have developed novel (and superficially
different) techniques.  Sch\"urmann and Poswolsky's approach seems
particularly similar to ours; they distinguish between variables and
parameters (names), and use ordered contexts with a restriction
operation similar to ours.  Each of them is considerably more
complicated than \DNTT, while sharing the advantages of higher-order
abstract syntax.  Poulliard and Pottier~\cite{poulliard10icfp}
recently proposed an interface in Agda which can be implemented either
using nominal terms or de Bruijn terms.  This approach may provide a
starting point for encoding a \DNTT-like language in Agda or Coq,
analogous to Harper and Licata's embedding of higher-order abstract
syntax.  It is a compelling open question how to relate these
techniques to nominal techniques (and to each other).  Developing such
encodings for nominal and various higher-order approaches in a common
metalanguage could be a way to compare their expressiveness.

\section{Dependent Nominal Type Theory}\labelSec{nlf}

The syntax of \DNTT is a straightforward extension of that of LF.  We
fix countable, disjoint sets of variables $x,y$, names $\Aa,\Ab$, 
object constants $c,d$, type constants $a,b$, and name-type constants
$\alpha,\beta$. The syntactic classes comprise  \emph{objects},
\emph{type families} (or just \emph{types}) which classify objects,
and \emph{kinds} which classify types.  The syntax of \DNTT kinds,
types, and objects is as follows:
\[
\begin{array}{rcll}
  K & ::= & \type \mid \PI{x}{A}{K}\\
  A,B & ::= & a \mid 
  A~M \mid \PI{x}{A}{B}~& \|~ \alpha \mid \NEW{\Aa}{\alpha}{B}\\
  M,N &::=& c \mid x \mid \LAM{x}{A}{M} \mid M~N ~&\|~ \Aa \mid \ABS{\Aa}{\alpha}{M} \mid M \conc \Aa
\end{array}
\]
We omit type-level lambda-abstraction, as it complicates the
metatheory yet does not add any expressive power to
LF~\cite{geuvers99mscs}.  The new syntactic cases of \DNTT are
distinguished using two parallel bars ($\|$).  As in LF, kinds include
$\type$, the kind of all types, and dependent kinds $\Pi x{:}A.K$ that
classify type families.  Types include constants $a$, applications
$A~M$ of type constructors to term arguments, and dependent types $\Pi
x{:}A.B$.  Name-types $\alpha$ are constants and thus cannot depend on
objects.  We include a dependent name-abstraction type constructor,
$\NEW{\Aa}{\alpha}{B}$, where $\alpha$ must be a name type.  Terms
include term constants $c$, variables $x$, applications $M~N$, and
$\lambda$-abstractions $\lambda x{:}A.M$ as in LF. In addition, terms
include names $\Aa$, name-abstractions $\ABS{\Aa}{\alpha}{M}$, and
name-applications $M \conc \Aa$ (also known as concretions).  Note
that the name argument of a concretion must be a literal name, not an
arbitrary term.  We adopt the same precedence conventions for
abstractions and concretions as for $\lambda$-abstraction and
application. For example, $\ABS{\Aa}{\alpha}{M\conc \Aa} =
\ABS{\Aa}{\alpha}{(M \conc \Aa)}$, not $(\ABS{\Aa}{\alpha}{M})\conc
\Aa$.

The $\Pi$ type constructor and $\lambda$ term constructor bind
variables in the usual way.  The $\NEW{\Aa}{\alpha}{B}$ type constructor
and  $\ABS{\Aa}{\alpha}{M}$ term constructor bind the name $\Aa$ in $B$ or
$M$ respectively, so are subject to $\alpha$-renaming.  The functions
$FV(-)$ and $FN(-)$ compute the set of free variables or free names of
a kind, type, or object; we write $FVN(-)$ for $FV(-) \cup FN(-)$.  As
in LF, when $x \not\in FV(B)$, we write $\PI{x}{A}{B}$ as the function
type $A \to B$; similarly, if $\Aa \not\in FN(B)$, we write
$\NEW{\Aa}{\alpha}{B}$ as the name-abstraction type $\Abs{\alpha}{B}$.  We
employ simultaneous substitutions $\theta$ of the form
\[\theta ::= \cdot \mid \theta,M/x \mid \theta,\Aa/\Ab\]
By convention, a substitution assigns at most one expression/name to
each variable/name.  We write $\theta(x)$ or $\theta(\Aa)$ for the
expression which $\theta$ assigns to $x$ or $\Aa$
respectively. Simultaneous substitution application $M[\theta]$ is
defined in \refFig{substitution}.

\begin{figure}[tb]
\[\begin{array}{rcll}
\type[\theta] &=& \type\\
(\PI{x}{A}{K})[\theta] &=& \PI{x}{A[\theta]}{K[\theta]} &\quad (x \not\in FV(\theta))
\medskip\\
a[\theta] &=& a\\
 \alpha[\theta] &=& \alpha\\
 (A~M)[\theta] &=& A[\theta] ~ M[\theta]\\
(\PI{x}{A}{B})[\theta] &=& \PI{x}{A[\theta]}{B[\theta]} &\quad (x \not\in FV(\theta))\\
(\NEW{\Aa}{\alpha}{B})[\theta] &=& \NEW{\Aa}{\alpha}{B[\theta]}&\quad (\Aa \not\in FN(\theta))
\medskip\\
c[\theta] &=& c\\
x[\theta] &=& \theta(x)\\
(\LAM{x}{A}{M})[\theta] &=& \LAM{x}{A[\theta]}{M[\theta]} &\quad (x \not\in FV(\theta))\\
(M~N)[\theta] &=& M[\theta] ~ N[\theta]\\
\Aa[\theta] &=& \theta(\Aa)\\
(\ABS{\Aa}{\alpha}{M})[\theta] &=& \ABS{\Aa}{\alpha}{M[\theta]} &\quad (\Aa \not\in FN(\theta))\\
(M\conc \Aa)[\theta] &=& M[\theta] \conc \Aa[\theta]
\medskip\\
\cdot[\theta] &=& \cdot\\
(\sigma,M/x)[\theta] &=& \sigma[\theta],M[\theta]/x
\end{array}\]
\caption{Substitution application and composition}\labelFig{substitution}
\end{figure}

\begin{figure*}[h]
\[\begin{array}{c}
\infer{\sigwf{\cdot}}{}\qquad
\infer{\sigwf{\Sigma,a{:}K}}{\kdwf{\cdot}{K} & \sigwf{\Sigma}}\qquad
\infer{\sigwf{\Sigma,c{:}A}}{\tywf{\cdot}{A}{\type} & \sigwf{\Sigma}}\qquad
\infer{\sigwf{\Sigma,\alpha{:}\nametype}}{\sigwf{\Sigma}}
\smallskip\\
\infer{\ctxwf{\cdot}}{}\qquad 
\infer{\ctxwf{\Gamma,\bind{x}{A}}}{ \tywf{\Gamma}{A}{\type} & \ctxwf{\Gamma}}
\qquad
\infer{\ctxwf{\Gamma\#\bind{\Aa}{\alpha}}}{\alpha:\nametype \in \Sigma
  & \ctxwf{\Gamma}}
\end{array}\]
\caption{\DNTT well-formedness rules: signatures, contexts}\labelFig{wf-sig-ctx}
\[\begin{array}{c}
\infer[\mathsf{type\_k}]{\kdwf{\Gamma}{\type}}{}\qquad
\infer[\mathsf{pi\_k}]{\kdwf{\Gamma}{\Pi x{:}A.K}}{\tywf{\Gamma}{A}{\type} & \kdwf{\Gamma,\bind{x}{A}}{K}}
\end{array}\]
\caption{\DNTT well-formedness rules: kinds}\labelFig{wf-kind}
\[\begin{array}{c}
\infer[\mathsf{con\_t}]{\tywf{\Gamma}{a}{K}}{a{:}K \in \Sigma}
\qquad
\infer[\mathsf{app\_t}]{\tywf{\Gamma}{A~M}{K[M/x]}}{\tywf{\Gamma}{A}{\Pi x{:}B.K}  & \tywf{\Gamma}{M}{B}}
\smallskip\\
\infer[\mathsf{pi\_t}]{\tywf{\Gamma}{\Pi x{:}A.B}{\type}}{\tywf{\Gamma}{A}{\type} & \tywf{\Gamma,\bind{x}{A}}{B}{\type}}
\qquad
\infer[\mathsf{conv\_t}]{\tywf{\Gamma}{A}{K}}{\tywf{\Gamma}{A}{K'} &
  \kdwf{\Gamma}{K=K'}}
\smallskip\\
\infer[\mathsf{nm\_t}]{\tywf{\Gamma}{\alpha}{\type}}{\alpha:\nametype \in \Sigma}
\quad
\infer[\mathsf{new\_t}]{\tywf{\Gamma}{\NEW{\Aa}{\alpha}{B}}{\type}}{\alpha:\nametype \in \Sigma & \tywf{\Gamma\#\bind{\Aa}{\alpha}}{B}{\type}}
\end{array}\]
\caption{\DNTT well-formedness rules: type families}\labelFig{wf-type}
\[\begin{array}{c}
\infer[\mathsf{con\_o}]{\wf{\Gamma}{c}{A}}{c{:}A \in \Sigma}
\quad
\infer[\mathsf{var\_o}]{\wf{\Gamma}{x}{A}}{\bind{x}{A} \in \Gamma} 
\quad
\infer[\mathsf{nm\_o}]{\wf{\Gamma}{\Aa}{\alpha}}{\Aa {:} \alpha \in \Gamma}
\smallskip\\
\infer[\mathsf{lam\_o}]{\wf{\Gamma}{\lambda x{:}A.M}{\Pi x{:}A.B}}{ \tywf{\Gamma}{A}{\type} & \wf{\Gamma,\bind{x}{A}}{M}{B}}
\qquad
\infer[\mathsf{app\_o}]{\wf{\Gamma}{M~N}{B[N/x]}}{\wf{\Gamma}{M}{\Pi
    x{:}A.B}  & \wf{\Gamma}{N}{A}}
\smallskip\\
\infer[\mathsf{abs\_o}]{\wf{\Gamma}{\ABS{\Aa}{\alpha}{M}}{\NEW{\Aa}{\alpha}{B}}}{\alpha:\nametype \in \Sigma & \wf{\Gamma\#\Aa{:}\alpha}{M}{B}}
\quad
\infer[\mathsf{conc\_o}]{\wf{\Gamma}{M \conc \Aa}{B}}{
  \wfres{\Gamma}{a}{\alpha}{\Gamma'} &
  \wf{\Gamma'}{M}{\NEW{\Aa}{\alpha}{B}}}
\end{array}\]
\caption{\DNTT well-formedness rules: objects}\labelFig{wf-object}
\end{figure*}
\begin{figure*}[tb]
\[
  \infer[\mathsf{res\_id}]{\wfres{\Gamma\#\bind{\Aa}{\alpha}}{\Aa}{\alpha}{\Gamma}}{}
\quad
\infer[\mathsf{res\_nm}]{\wfres{\Gamma\#\bind{\Ab}{\beta}}{\Aa}{\alpha}{\Gamma'\#\bind{\Ab}{\beta}}}{ \wfres{\Gamma}{\Aa}{\alpha}{\Gamma'}}
\quad
\infer[\mathsf{res\_var}]{\wfres{\Gamma,\bind{x}{A}}{\Aa}{\alpha}{\Gamma'}}{\wfres{\Gamma}{\Aa}{\alpha}{\Gamma'}}
\]
\[\cdot-\Aa = \cdot\qquad
(\theta,M/x)-a = \theta-\Aa\qquad
(\theta,\Aa'/\Aa)-\Aa = \theta \qquad
(\theta,\Ab'/\Ab)-\Aa = (\theta-\Aa),\Ab'/\Ab
\]
\caption{Context and substitution restriction}\labelFig{restriction}
\end{figure*}

\begin{figure*}[tb]
\[\begin{array}{c}
\quad
\infer[\mathsf{eq\_ext\_t}]{\wf{\Gamma}{A_1=A_2}{\Pi
    x:B.K}}{\wf{\Gamma,x:B}{A_1~x=A_2~x}{K}}
\quad
\infer[\mathsf{eq\_new\_t}]{\wf{\Gamma}{\NEW{\Aa}{\alpha}{A}=\NEW{\Aa}{\alpha}{B}}{\type}}{\wf{\Gamma\#\bind{\Aa}{\alpha}}{A=B}{\type}}
\smallskip\\
\infer[\mathsf{eq\_nm}]{\wf{\Gamma}{\Aa=\Aa}{\alpha}}{\Aa{:}\alpha \in
  \Gamma}
\smallskip\\
\infer[\mathsf{eq\_abs}]{\wf{\Gamma}{\abs{\Aa{:}\alpha}{M}=\abs{\Aa{:}\alpha}{N}}{\new \Aa{:}\alpha.A}}{
  \wf{\Gamma\#\Aa{:}\alpha}{M=N}{A}}
\quad
\infer[\mathsf{eq\_conc}]{\wf{\Gamma}{M\conc \Ab = N \conc \Ab}{A[\Ab/\Aa]}}{
  \wfres{\Gamma}{\Ab}{\alpha}{\Gamma'} & 
  \wf{\Gamma'}{M=N}{\new \Aa{:}\alpha.A}}
\smallskip\\
\infer[\mathsf{eq\_nm\_beta}]{\wf{\Gamma}{(\abs{\Aa{:}\alpha}M)\conc \Ab = N[\Ab/\Aa]}{A[\Ab/\Aa]}}{
  \wfres{\Gamma}{\Ab}{\alpha}{\Gamma'} 
  & \wf{\Gamma'\#\Aa{:}\alpha}{M=N}{A}}
\quad
\infer[\mathsf{eq\_nm\_eta}]{\wf{\Gamma}{M=N}{\new \Aa{:}\alpha.A}}{\wf{\Gamma\#\Aa:\alpha}{M\conc \Aa = N\conc \Aa}{A}}
\end{array}\]
\caption{New definitional equivalence rules of $\DNTT$}\labelFig{defeq-rules}
\end{figure*}

As in LF, the language of constants used in a specification is
described by a \emph{signature} assigning (closed) kinds to type
constants and (closed) types to object constants.  The contexts
$\Gamma$ used in \DNTT are also similar to those of LF, except that
bindings of names introduced by $\new$ are written $\Gamma\#\Aa{:}\alpha$,
to indicate that such names must be ``fresh'' for the rest of the
context:
\[\begin{array}{rcll}
  \Sigma &::=& \cdot \mid \Sigma,c{:}A\mid \Sigma,a{:}K~&\|~ \Sigma,\alpha{:}\nametype \\
  \Gamma &::=& \cdot \mid \Gamma,x{:}A~&\|~ \Gamma\#\Aa{:}\alpha
\end{array}\]
By convention, the constants and variables on the left-hand side of
`${:}$' in a signature or context are always distinct.  This
implicitly constrains the inference rules.

We extend Harper and Pfenning's presentation of the LF typing and
equality rules~\cite{harper05tocl}.  All judgments except signature
formation are implicitly parametrized by a signature $\Sigma$.  We omit
explicit freshness and signature or context well-formedness
constraints.

The well-formedness rules of \DNTT are shown in
Figures~\ref{fig:wf-sig-ctx}--\ref{fig:wf-object}.  The additional
definitional equivalence rules of \DNTT are shown in
\refFig{defeq-rules}. We omit the standard definitional equivalence
rules of LF; we add a type-level extensionality rule that was omitted
from Harper and Pfenning's presentation but is
admissible~\cite{urban11tocl}.  The new rules define the behavior of
names and name-abstraction or $\new$-types.  The $\new$-type formation
rule is similar to the $\Pi$-type formation rule, except using the
$\Gamma\#\Aa{:}\alpha$ context former.  The rule for name-abstraction
is similar.  In the rule for concretion, the name at which the
abstraction term is instantiated is \emph{removed} from the context
using a \emph{context restriction} judgment
$\wfres{\Gamma}{\Aa}{\alpha}{\Gamma'}$, shown in \refFig{restriction}.
This judgment states that $\Aa : \alpha$ is bound in $\Gamma$ and
$\Gamma'$ is the result of removing the name $\Aa$ from $\Gamma$,
along with any variables that were introduced more recently than
$\Aa$.  For technical reasons, we also need a \emph{substitution
  restriction} operation $\theta - \Aa$, also shown in
\refFig{restriction}.

The use of an explicit context restriction judgment is a key
difference between \DNTT and other systems that use bunched contexts,
such as Sch\"opp and Stark's system~\cite{schoepp06phd,schoepp04csl}
or O'Hearn and Pym's Logic of Bunched Implications~\cite{ohearn99bsl}.
In those theories, context conversion steps can be performed
nondeterministically at any point.  This complicates
equivalence-checking in the presence of dependent types, because we
have to be careful to ensure that context conversion steps do not make
the context ill-formed.  In \DNTT, we constrain the use of bunched
contexts so that standard typechecking and equivalence algorithms for
LF can be re-used with minimal changes.

We consider a substitution to be well-formed (written
$\wfsub{\Gamma}{\theta}{\Gamma'}$) when it maps the variables and
names of some context $\Gamma'$ to terms and names well-formed with
respect to another context $\Gamma$, while respecting the freshness
requirements of $\Gamma'$.  This is formalized as follows:
\[\begin{array}{c}
\infer{\wfsub{\Gamma}{\cdot}{\cdot}}{} \qquad
\infer{\wfsub{\Gamma}{\theta,M/x}{\Gamma',x{:}A}}{\wfsub{\Gamma}{\theta}{\Gamma'} & \wf{\Gamma}{M}{A[\theta]}}\smallskip\qquad
\infer{\wfsub{\Gamma}{\theta,\Aa/\Ab}{\Gamma'\#\Ab{:}\alpha}}{\wfres{\Gamma}{\Aa}{\alpha}{\Gamma''} & \wfsub{\Gamma''}{\theta}{\Gamma'}}
\end{array}\]
In addition, we consider a context $\Gamma'$ to be a \emph{subcontext}
of $\Gamma$ (written $\Gamma' \preceq \Gamma$) if
$\wfsub{\Gamma}{\id_{\Gamma'}}{\Gamma'}$ holds, where $\id_{\Gamma'}$
denotes the identity substitution on context $\Gamma'$.  Note that, for
example, $\cdot,x{:}A\#\Aa{:}\alpha \preceq \cdot\#\Aa{:}\alpha,x{:}A$
holds but not the converse, because the former context guarantees that
$\Aa$ is fresh for $x$ and the latter does not.

We employ a number of standard metatheoretic results about LF, which
extend to \DNTT without difficulty.  We next summarize some basic
metatheoretic properties of \DNTT.  Let $\JJ$ range over
well-formedness assertions $K : \kind$, $A : K$, $M : A$ or equality
assertions $K=K' : \kind$, $A = A':K$, $M = M':A$.

\begin{lem}[Determinacy of restriction]\labelLem{res-det}
  If $\wfres{\Gamma}{\Aa}{\alpha}{\Gamma_1}$ and
  $\wfres{\Gamma}{\Aa}{\alpha}{\Gamma_2}$ then $\Gamma_1 = \Gamma_2$.
\end{lem}
\begin{proof}
  Straightforward induction on the first derivation using inversion on the
  second.
\end{proof}
\begin{lem}[Restriction implies weakening]\labelLem{restriction-weakening}
  If $\wfres{\Gamma}{\Aa}{\alpha}{\Gamma'}$ then
  $\Gamma'\#\Aa{:}\alpha \preceq \Gamma$.
\end{lem}
\begin{proof}
  Straightforward, by induction on the structure of derivations.
\end{proof}

\begin{lem}[Weakening]\labelLem{weakening}
Suppose $\Gamma' \succeq \Gamma$.  Then (1) 
    If $\wfres{\Gamma}{\Aa}{\alpha}{\Gamma_0}$  then
    $\wfres{\Gamma'}{\Aa}{\alpha}{\Gamma_0'}$ for some $\Gamma_0'
    \succeq \Gamma_0$.
  (2) If $\judge{\Gamma}{\mathcal{J}}$ then $\judge{\Gamma'}{\mathcal{J}}$.
\end{lem}
\begin{lem}
  [Substitution restriction] \labelLem{subst-res}
If $\wf{\Gamma'}{\theta}{\Gamma}$ and
  $\wfres{\Gamma}{\Aa}{\alpha}{\Gamma_0}$ then
  $\wfres{\Gamma'}{\theta(\Aa)}{\alpha}{\Gamma_0'}$ and
  $\wf{\Gamma_0'}{\theta-\Aa}{\Gamma_0}$ for some $\Gamma_0'$.
\end{lem}
\begin{lem}[General Substitution]\labelLem{general-subst}
  Assume $\judge{\Gamma}{\JJ}$ and
  $\wfsub{\Gamma'}{\theta}{\Gamma}$.  Then
  $\judge{\Gamma'}{\JJ[\theta]}$.
\end{lem}
\begin{proof}
  The cases for existing LF rules are straightforward.  Of the new
  cases, only the rule for concretion is interesting.  Suppose we have
\[\infer{\wf{\Gamma}{M\conc \Ab}{B[\Ab/\Aa]}}{
\deduce{\wfres{\Gamma}{\Ab}{\alpha}{\Gamma_0}}{\DD_1} &
 \deduce{\wf{\Gamma_0}{M}{\new \Aa{:}\alpha.B}}{\DD_2}}\]
Then by assumption, we have $\wfsub{\Gamma'}{\theta}{\Gamma}$.  Using
\refLem{subst-res} on $\DD_1$, we have
$\wfres{\Gamma'}{\Ab[\theta]}{\alpha}{\Gamma_0'}$ and
$\wf{\Gamma_0'}{\theta-\Ab}{\Gamma_0}$ for some $\Gamma_0'$.  Thus, by
induction, $\wf{\Gamma_0'}{M[\theta-\Ab]}{(\new
  \Aa{:}\alpha.B)[\theta-\Ab]}$ and by definition,
$\wf{\Gamma_0'}{M[\theta-\Ab]}{\new \Aa{:}\alpha.B[\theta-\Ab]}$.
Moreover, we may derive
\[
\infer{\wf{\Gamma'}{M[\theta-\Ab]\conc \Ab[\theta]}{B[\theta-\Ab][\Ab[\theta]/\Aa]}}{ 
\wfres{\Gamma'}{\Ab[\theta]}{\alpha}{\Gamma_0'}
&
\wf{\Gamma_0'}{M[\theta-\Ab]}{\new \Aa{:}\alpha.B[\theta-\Ab]}
}
\]
To conclude, we observe that $M[\theta-\Ab] = M[\theta]$ and
$B[\theta-\Ab][\Ab[\theta]/\Aa] = B[\Ab/\Aa][\theta]$ since the extra
variables and names mentioned in $\theta$ cannot be mentioned in $M$
or $B$.  So $\wf{\Gamma'}{(M\conc \Ab)[\theta]}{B[\Ab/\Aa][\theta]}$.
\end{proof}

\begin{cor}[Substitution]\labelCor{subst}
  If $\judge{\Gamma}{M:A}$ and $\judge{\Gamma,x{:}A}{\JJ}$,
  then $\judge{\Gamma}{\JJ[M/x]}$.
\end{cor}
\begin{proof}
Follows from \refLem{general-subst}, using $\theta = \id_\Gamma,M/x$,
which is easily seen to satisfy $\wfsub{\Gamma}{\id_\Gamma,M/x}{\Gamma,x{:}A}$.
\end{proof}
\begin{cor}[Renaming]\labelCor{renaming}
  If $\wfres{\Gamma}{\Aa}{\alpha}{\Gamma'}$ and $\judge{\Gamma'\#\Ab{:}\alpha}{\JJ}$,
  then $\judge{\Gamma}{\JJ[\Aa/\Ab]}$.
\end{cor}
\begin{proof}
  Follows from \refLem{general-subst}, using $\theta =
  \id_{\Gamma'},\Aa/\Ab$,
which satisfies $\wfsub{\Gamma}{\id_{\Gamma'},\Aa/\Ab}{\Gamma'\#\Ab{:}\alpha}$.
\end{proof}


As an initial check that these rules are sensible, we verify the
\emph{local soundness} and \emph{completeness} properties expressing
that typability is preserved by $\beta$-reduction and $\eta$-expansion
steps.  For $\beta$-reductions of name-abstractions, given
\[
\begin{array}{c}
\infer{\wf{\Gamma}{(\ABS{\Aa}{\alpha}{{M}}) \conc \Ab}{B(\Ab)}}
{\wfres{\Gamma}{\Ab}{\alpha}{\Gamma'} & \infer{\wf{\Gamma'}{\ABS{\Aa}{\alpha}{{M}}}{\NEW{\Aa}{\alpha}{B(\Aa)}}}
{\wf{\Gamma'\#\Aa{:}\alpha}{M}{B(\Aa)}}}
\end{array}
\]
we conclude that $ \wf{\Gamma}{M[\Ab/\Aa]}{B(\Ab)}$ by
\refCor{renaming}.  For $\eta$-expansion of name-abstractions, given
a derivation of $\wf{\Gamma}{M}{\NEW{\Aa}{\alpha}{B}}$, and $\Aa \not\in
\Gamma$, we can expand to:
\[ 
\infer{\wf{\Gamma}{\ABS{\Aa}{\alpha}{M\conc \Aa}}{\NEW{\Aa}{\alpha}{B}}}
{\infer{\wf{\Gamma\#\Aa{:}\alpha}{M\conc \Aa}{B}}
{\wfres{\Gamma\#\Aa{:}\alpha}{\Aa}{\alpha}{\Gamma} & \wf{\Gamma}{M}{\NEW{\Aa}{\alpha}{B}}}
}
\]

As further examples of the properties of \DNTT, observe that for any
$A,B$ with $\Aa \not\in FN(B)$ we have ``weakening'' and ``exchange''
properties for $\new$:
\[
\begin{array}{l}
\wf{}{\lambda x{:}B. \ABS{\Aa}{\alpha}{x}}{B \to
  \NEW{\Aa}{\alpha}{B}} \;,\\
\wf{}{\lambda
  x{:}(\NEW{\Aa}{\alpha}{\NEW{\Ab}{\beta}{A})}. \abs{\Ab{:}\beta}{\abs{\Aa{:}\alpha}{x\conc
      \Aa \conc \Ab}}}{\NEW{\Aa}{\alpha}{\NEW{\Ab}{\beta}{A}} \to
 \NEW{\Ab}{\beta}{ \NEW{\Aa}{\alpha}{A}}}\;.
\end{array}
\]
 We might expect an inverse ``strengthening'' property, that is, 
$\new \Aa{:}\alpha.B \to B$, but this does not hold in general. The
following derivation gets stuck because there is no name $\Aa$ to
which to apply $x$:
\[
\infer{\wf{}{\LAM{x}{(\NEW{\Aa}{\alpha}{B})}{??}}{\NEW{\Aa}{\alpha}{B} \to B}}
{\wf{\bind{x}{\NEW{\Aa}{\alpha}{B}}}{ ??}{ B}}
\]
This makes sense, semantically speaking, because for example there is
no equivariant function from the nominal set
$\abs{\mathbb{A}}{\mathbb{A}}$ to $\mathbb{A}$ (where $\mathbb{A}$ is
a set of names).  We will not develop a nominal set semantics of \DNTT
here, but such a semantics was developed for a simply-typed calculus
in~\cite{cheney08lfmtp}.

There are natural functions that are definable in the nominal set
semantics that are not definable in \DNTT.  Suppose we have a function
$h : \mathbb{A} \times X \to Y$ such that for any name $\Aa$, if $\Aa$
is fresh for $x$ then $\Aa$ is fresh for $h(\Aa,x)$.  Then, as
discussed by Pitts~\cite{pitts06jacm}, we can define a function $h' :
\abs{\mathbb{A}}{X} \to Y$ satisfying $h(\Aa,x) = h'(\abs{\Aa}{x})$.
(This function is obtained by lifting $h$ to equivalence classes of
name-abstractions; the freshness condition for $h$ is sufficient to
ensure that $h$ respects $\alpha$-equivalence classes.)

As a simple example, suppose for the moment we include a standard
$\textrm{option}$ type and consider the function $g' :
\abs{\mathbb{A}}{\mathbb{A}} \to \mathbb{A}~\option$ defined by
\[g(\Aa,x) = \left\{\begin{array}{ll}
\NONE & \Aa = x\\
\SOME(x) & \Aa \neq x
\end{array}\right.
\]
This function lets us test whether an abstraction is of the form
$\abs{\Aa}{\Aa}$, and if it is not, extracts the body.  We have $\Aa
\fresh x $ implies $\Aa \fresh g(\Aa,x)$, but $g'$ cannot be defined
as a \DNTT term $N : \abs{\alpha}{\alpha} \to \alpha~\option$.  As
another example, consider the function
$k':\abs{\mathbb{A}}{\mathbb{N}} \to \mathbb{N}$ obtained from
$k(\Aa,n) = n$.  We can obviously define a natural number type $\nat$
in \DNTT, but we cannot define a \DNTT function $M :
\Abs{\alpha}{\nat} \to \nat$ satisfying $M(\abs{\Aa}{n}) = n$.

In \DNTT, we currently have no general way to define such functions, and
it is not immediately obvious how to accommodate them.  One
possibility might be to add a term constructor $\nu \Aa{:}\alpha.M$
with well-formedness rule:
\[
\infer{\wf{\Gamma}{\nu
    \Aa{:}\alpha.M}{A}}{\wf{\Gamma\#\Aa{:}\alpha}{M}{A} &
  \text{``$\Aa$ fresh for $M$''}}
\]
Roughly this approach (without the freshness side-condition) is taken
in a simply-typed calculus called Nominal System T~\cite{pitts10popl,pitts11jfp}.
However, there are significant complications with incorporating this
approach to name-restriction into a dependent type theory, explored
further in \refSec{discussion}.

In addition to the basic results presented so far, we need to
establish a number of straightforward properties for \DNTT, including
validity, inversion, and injectivity for $\Pi$ and $\new$.  These
properties (and their proofs) are essentially the same as for LF as
given in~\cite{harper05tocl,urban11tocl} and are omitted.

\section{Equivalence and canonical forms}\labelSec{meta}

\begin{figure*}[tb]
  \centering
  \[
  \begin{array}{c}
\infer{M~N \whr M'~N}{M \whr N}\quad
\infer{(\lambda x:A.M)~N \whr M[N/x]}{}\quad
\infer{M\conc \Aa\whr N \conc \Aa}{M\whr N} \quad
\infer{(\abs{\Aa{:}\alpha}{M})\conc \Ab\whr M[\Ab/\Aa]}{}
 \end{array}
  \]
  \caption{Weak head reduction}\labelFig{whr}
\[
\begin{array}{c}
\infer{\algeq{\Delta}{M}{N}{a^-}}{M \whr M' &
  \algeq{\Delta}{M'}{N}{a^-}}
\qquad
\infer{\algeq{\Delta}{M}{N}{a^-}}{N \whr N' &
  \algeq{\Delta}{M}{N'}{a^-}}
\smallskip\\
\infer{\algeq{\Delta}{M}{N}{a^-}}{\streq{\Delta}{M}{N}{a^-}}
\qquad
\infer{\algeq{\Delta}{M}{N}{\tau_1 \to
    \tau_2}}{\algeq{\Delta,x:\tau_1}{M~x}{N~x}{\tau_2}}
\qquad
\infer{\algeq{\Delta}{M}{N}{\abs{\alpha}{\tau}}}{\algeq{\Delta\#\bind{\Aa}{\alpha}}{M\conc \Aa}{N \conc \Aa}{\tau}}
\medskip\\
\infer{\streq{\Delta}{x}{x}{\tau}}{ \bind{x}{\tau}\in \Delta}
\qquad
\infer{\streq{\Delta}{c}{c}{A^-}}{ \bind{c}{A}\in \Sigma}
\qquad
    \infer{\streq{\Delta}{\Aa}{\Aa}{\alpha}}{\bind{\Aa}{\alpha}\in\Delta}
\smallskip\\
\infer{\streq{\Delta}{M_1~M_2}{N_1~N_2}{\tau_2}}{
  \streq{\Delta'}{M_1}{N_1}{\tau_1 \to \tau_2} & \algeq{\Delta}{M_2}{N_2}{\tau_1}}
\qquad 
\infer{\streq{\Delta}{M\conc \Aa}{N\conc
    \Aa}{\tau}}{\wfres{\Delta}{\Aa}{\alpha}{\Delta'} &
  \streq{\Delta'}{M}{N}{\abs{\alpha}{\tau}}}
\end{array}
\]\caption{Algorithmic and structural equivalence rules for objects}\labelFig{alg-equiv}
\[
\begin{array}{c}
\infer{\algeq{\Delta}{A}{B}{\type^-}}{\streq{\Delta}{A}{B}{\type^-}}
\quad 
\infer{\algeq{\Delta}{A}{B}{\tau \to \kappa}}{\algeq{\Delta,x{:}\tau}{A~x}{B~x}{\kappa}}
\smallskip\\
\infer{\algeq{\Delta}{\Pi x{:}A_1.A_2}{\Pi    x{:}B_1.B_2}{\type^-}}{\algeq{\Delta}{A_1}{B_1}{\type^-} & \algeq{\Delta,x:A_1^-}{A_2}{B_2}{\type^-}}
\qquad
\infer{\algeq{\Delta}{\NEW{\Aa}{\alpha}{B}}{\NEW{\Aa}{\alpha}{B'}}{\type^-}}{\algeq{\Delta\#\Aa{:}\alpha}{B}{B'}{\type^-}}
\smallskip\\
\infer{\streq{\Delta}{a}{a}{K^-}}{\bind{a}{K} \in \Sigma}
\qquad
\infer{\streq{\Delta}{\alpha }{\alpha}{\type^-}}{\alpha:\nametype \in \Sigma}
\qquad
\infer{\streq{\Delta}{A~M}{B~N}{\kappa}}{\streq{\Delta}{A}{B}{\tau
      \to \kappa} & \algeq{\Delta}{M}{N}{\tau}}
\end{array}
\]
\caption{Algorithmic and structural equivalence rules for types}\labelFig{alg-equiv-type}
\[
\begin{array}{c}
\infer{\algeq{\Delta}{\type}{\type}{\kind^-}}{} \qquad
\infer{\algeq{\Delta}{\Pi x{:}A.K}{\Pi
    x{:}B.L}{\kind^-}}{\algeq{\Delta}{A}{B}{\type^-} & \algeq{\Delta,x:A^-}{K}{L}{\kind^-}}
\end{array}
\]
\caption{Algorithmic equivalence rules for
  kinds}\labelFig{alg-equiv-kind}
\end{figure*}

In this section we show that the definitional equivalence and
well-formedness judgments of \DNTT are decidable.  In previous
work~\cite{cheney08lfmtp}, we showed strong normalization for a
simply-typed lambda calculus with names and name-abstraction types by
translating name-types to function types and re-using standard results
for the simply-typed lambda calculus.  Here, we prove the desired
results directly, based on Harper and Pfenning's decidability
proof~\cite{harper05tocl}.

Harper and Pfenning's approach is based on an algorithmic equivalence
judgment that weak head-normalizes LF terms.  The judgment only tracks
simple types $\tau$ for variables and terms may not necessarily be
well-formed.  The algorithm is shown sound and complete for
well-formed LF terms with respect to the definitional equivalence
rules.  Soundness is proved syntactically, whereas completeness
involves a logical relation argument.  The logical relation is defined
by induction on the structure of simple types.

We extend their simple types and kinds with name-abstraction types as
follows:
\[\tau ::= a^-  \mid \tau \to \tau' ~\|~ \abs{\alpha}{\tau}\mid
\alpha \qquad \kappa ::= \type^- \mid \tau \to \kappa\]
and extend the erasure function by defining $(\alpha)^- = \alpha $ and
$(\NEW{\Aa}{\alpha}{A})^- = \abs{\alpha}{A^-}$.  We consider simple
contexts $\Delta$ mapping variables to simple types.  We extend the
weak head reduction and algorithmic equivalence judgments with rules
for names and name-abstractions (\refFig{alg-equiv}).  Also, we define a restriction judgment
$\wfres{\Delta}{\Aa}{\alpha}{\Delta'}$ for simple contexts; its
definition is identical to that for dependently-typed contexts and so
is omitted.

There are a number of additional properties of erasure and
algorithmic equivalence that are needed for the following soundness
and completeness results, but again these are essentially the same as
in~\cite{harper05tocl,urban11tocl} so are omitted.

\subsection{Soundness}

The proof of soundness is syntactic.  Note however that we include a
rule for type-level extensionality, avoiding a subtle problem in
Harper and Pfenning's presentation (see~\cite[sec. 3.4]{urban11tocl}).

\begin{thm}[Subject reduction]\labelThm{sr}
  If $\wf{\Gamma}{M}{A}$ and $M \whr M'$ then $\wf{\Gamma}{M=M'}{A}$
  (and hence $\wf{\Gamma}{M'}{A}$ also).
\end{thm}
\begin{proof}
  By induction on the derivation of $M \whr M'$, with most cases standard.
  \begin{iteMize}{$\bullet$}
  \item If the derivation is of the form:
    \[
    \infer{M \conc \Aa \whr M' \conc \Aa}{M \whr M'}
    \]
    then by inversion we must have
    $\wfres{\Gamma}{\Aa}{\alpha}{\Gamma'}$ and $\wf{\Gamma'}{M}{\new
      \Ab{:}\alpha.A'}$ where $\wf{\Gamma'}{A =
      A'[\Aa/\Ab]}{\type}$. Hence, by induction we know that
    $\wf{\Gamma'}{M=M'}{\new \Ab{:}\alpha.A'}$, and we may derive
    \[\infer{\wf{\Gamma}{M\conc \Aa = M' \conc
        \Aa}{A}}{
      \infer{\wf{\Gamma}{M\conc \Aa = M' \conc
        \Aa}{A'[\Aa/\Ab]}}{\wfres{\Gamma}{\Aa}{\alpha}{\Gamma'} &
      \wf{\Gamma'}{M=M'}{\new \Ab{:}\alpha.A'}}
  & \wf{\Gamma'}{A =
      A'[\Aa/\Ab]}{\type}}
    \]
  \item If the derivation is of the form:
    \[
    \infer{(\abs{\Aa{:}\alpha}{M})\conc \Ab \whr M[\Ab/\Aa]}{}
    \]
    then by inversion we must have
    $\wfres{\Gamma}{\Ab}{\alpha}{\Gamma'}$ and
    $\wf{\Gamma'}{\abs{\Aa{:}\alpha}{M}}{\new \Aa{:}\alpha.A'}$, where
    $\wf{\Gamma}{A=A'[\Ab/\Aa]}{\type}$.  Moreover, again by inversion
    we must have $\wf{\Gamma'\#\Aa{:}\alpha}{M}{A''}$ where
    $\wf{\Gamma'\#\Aa{:}\alpha}{A'=A''}{\type}$.  Thus, we may derive:
  \[
   \infer{\wf{\Gamma}{(\abs{\Aa{:}\alpha}{M})\conc \Ab = M[\Ab/\Aa]}{A'[\Ab/\Aa]}}
    {\wfres{\Gamma}{\Ab}{\alpha}{\Gamma'}
      &
      \infer{\wf{\Gamma'\#\Aa{:}\alpha}{M = M}{A'}}
      {\infer{\wf{\Gamma'\#\Aa{:}\alpha}{M}{A'}}
        {\wf{\Gamma'\#\Aa{:}\alpha}{M}{A''}  
          &
          \wf{\Gamma'\#\Aa{:}\alpha}{A'=A''}{\type}
        }
      }
    }
 \]
 Since $\wf{\Gamma}{A=A'[\Ab/\Aa]}{\type}$, we can conclude
 $\wf{\Gamma}{(\abs{\Aa{:}\alpha}{M})\conc \Ab = M[\Ab/\Aa]}{A}$, as
 desired.
  \end{iteMize}
\end{proof}

\begin{lem}[Soundness of restriction]
  If $\Gamma,\Gamma_0$ are well-formed and $\wfres{\Gamma^-}{\Aa}{\alpha}{\Gamma_0^-}$ then
  $\wfres{\Gamma}{\Aa}{\alpha}{\Gamma_0}$.
\end{lem}
\begin{proof}
  Straightforward induction on derivations.
\end{proof}\medskip

\begin{thm}[Soundness]\hfill
  \begin{enumerate}[\em(1)]\item 
    If $\algeq{\Gamma^-}{M}{N}{A^-}$ and $\wf{\Gamma}{M,N}{A}$ then
    $\wf{\Gamma}{M=N}{A}$.
  \item If $\streq{\Gamma^-}{M}{N}{\tau}$ and $\wf{\Gamma}{M}{A}$ and
    $\wf{\Gamma}{N}{B}$ then $\wf{\Gamma}{A=B}{\type}$ and
    $\wf{\Gamma}{M=N}{A}$ and $A^- =\tau = B^-$.
  \end{enumerate}
\end{thm}
\proof
  By simultaneous induction on the derivations of
  $\algeq{\Gamma^-}{M}{N}{A^-}$ and $\streq{\Gamma^-}{M}{N}{\tau}$.
  Again most cases are standard; we show the new cases only.
  \begin{iteMize}{$\bullet$}
  \item If the derivation is of the form:
    \[
    \infer{\streq{\Gamma^-}{\Aa}{\Aa}{\alpha}}{\bind{\Aa}{\alpha}\in\Gamma^-}
    \]
    then we must have $\Aa{:}\alpha \in \Gamma$ so we can conclude
    that $\wf{\Gamma}{\alpha=\alpha}{\type}$ and
    $\wf{\Gamma}{\Aa=\Aa}{\alpha}$ and $A^- = \alpha^- = \alpha =
    \alpha^- = B^-$.

  \item 
If the derivation is of the form:
    \[
    \infer{\algeq{\Gamma^-}{M}{N}{\abs{\alpha}{\tau}}}{
      \algeq{\Gamma^-\#\bind{\Aa}{\alpha}}{M\conc \Aa}{N\conc\Aa}{\tau}}
    \]
    where $A^- = \abs{\alpha}{\tau}$, then without loss of generality we
    assume $\Aa$ is fresh for $\Gamma,A,M,N$.  Then by inversion of
    erasure we must have $A = \NEW{\Ab}{\alpha}{A_0}$ for some $A_0$
    with $A_0^- = \tau$.  Without loss of generality, assume that
    $\Ab$ is fresh for $\Aa, \Gamma, A,M,N$. Moreover, we can easily
    show that $\wf{\Gamma\#\Aa{:}\alpha}{M\conc \Aa}{A_0[\Aa/\Ab]}$
    and similarly for $N$.  Then by induction, we know that
    $\wf{\Gamma\#\Aa{:}\alpha}{M\conc \Aa = N \conc
      \Aa}{A_0[\Aa/\Ab]}$, hence we can derive
\[
\infer{\wf{\Gamma}{M=N}{\NEW{\Aa}{\alpha}{A_0[\Aa/\Ab]}}}{\wf{\Gamma\#\Aa{:}\alpha}{M\conc\Aa
    = N \conc\Aa}{A_0[\Aa/\Ab]}}
\]
Since $A= \NEW{\Ab}{\alpha}{A_0}$ and $\Aa$ is sufficiently fresh, $A$
is $\alpha$-equivalent to $\NEW{\Aa}{\alpha}{A_0[\Aa/\Ab]}$, so $\wf{\Gamma}{M=N}{A}$, as desired.

  \item If the derivation is of the form:
    \[
    \infer{\streq{\Gamma^-}{M\conc \Aa}{N\conc \Aa}{\tau}}{
      \wfres{\Gamma^-}{\Aa}{\alpha}{\Gamma_0^-} &
      \streq{\Gamma_0^-}{M}{N}{\abs{\alpha}{\tau}}}
    \]
    then we know that $\wfres{\Gamma}{\Aa}{\alpha}{\Gamma_0}$ by the
    soundness of restriction.  Moreover, by inversion we know that
    $\wfres{\Gamma}{\Aa}{\alpha}{\Gamma_1}$ and $\wf{\Gamma_1}{M}{\new
      \Aa{:}\alpha.A_0}$ and $\wf{\Gamma}{A_0 = A}{\type}$ for some
    $\Gamma_1,A_0$, and similarly for $N$ for some $\Gamma_2,B_0$.  By
    determinacy of restriction (\refLem{res-det}) we know that $\Gamma_0 = \Gamma_1 =
    \Gamma_2$.  Hence, by induction we have that $\wf{\Gamma_0}{\new
      \Aa{:}\alpha.A_0 = \new \Aa{:}\alpha.B_0}{\type}$ and
    $\wf{\Gamma_0}{M=N}{\new \Aa{:}\alpha.A_0}$ and $(\new
    \Aa{:}\alpha.A_0)^- = \abs{\alpha^-}{\tau} = (\new
    \Aa{:}\alpha.B_0)^-$.  It follows immediately that $A_0^- = \tau =
    B_0^-$.  In addition, we have that
    $\wf{\Gamma_0\#\Aa{:}\alpha}{A_0 = B_0}{\type}$ by injectivity of
    $\new$-type equality.

    To conclude, we can derive:
    \[
     \infer{\wf{\Gamma}{A_0=B}{\type}}
      {\infer[W]{\wf{\Gamma}{A_0=B_0}{\type}}
        {\wf{\Gamma_0\#\Aa{:}\alpha}{A_0=B_0}{\type}}
        & \wf{\Gamma}{B_0=B}{\type}
      }
   \]
   where the inference labeled $W$ is by weakening since we must have
   $\Gamma_0\#\Aa{:}\alpha \preceq \Gamma$ by
   \refLem{restriction-weakening}. Next, observe that by transitivity
   we have $\wf{\Gamma}{A=B}{\type}$ since $\wf{\Gamma}{A=A_0}{\type}$
   holds.  Finally, we can also derive:
    \[
    \infer{\wf{\Gamma}{M\conc \Aa=N\conc \Aa}{A_0}}
    {
      \wfres{\Gamma}{\Aa}{\alpha}{\Gamma_0} & 
      \wf{\Gamma_0}{M=N}{\new \Aa{:}\alpha.A_0}
    }
    \]
\end{iteMize}
This completes the proof.
\qed

\subsection{Completeness}

The proof of completeness is by a Kripke logical relation argument.
The logical relation is extended with a case for name-abstraction
types in \refFig{logrel}. We first state the key properties of the
logical relations:

\begin{lem}[Logical substitution restriction]\labelLem{subst-restrict}
  Suppose that $\logrel{\Delta}{\theta=\sigma}{\Gamma^-}$ and
  $\wfres{\Gamma}{\Aa}{\alpha}{\Gamma_0}$.  Then $\theta(\Aa) =
  \sigma(\Aa)$ and there exists $\Delta_0$ such that
  $\wfres{\Delta}{\theta(\Aa)}{\alpha}{\Delta_0}$ and
    $\logrel{\Delta_0}{\theta-\Aa=\sigma-\Aa}{\Gamma_0^-}$.
\end{lem}
\begin{proof}
  It is straightforward to show that $\theta(\Aa) = \sigma(\Aa)$ by
  induction on the first derivation.  For the second part, the proof
  is by induction on the second derivation, using inversion and the
  definition of substitution restriction.
\end{proof}

\begin{figure*}[tb]
  \centering
  \begin{eqnarray*}
    \logrel{\Delta}{M=N}{\delta} & \iff & \algeq{\Delta}{M}{N}{\delta} \quad (\delta \in \{\alpha,a^-\})\\
    \logrel{\Delta}{M=N}{\tau_1\to\tau_2} &\iff& 
    \forall \Delta' \succeq \Delta. \logrel{\Delta'}{M'=N'}{\tau_1} 
    \impp \logrel{\Delta'}{M~M'=N~N'}{\tau_2}\\
    \logrel{\Delta}{M=N}{\abs{\alpha}{\tau}} & \iff & 
    \forall \Delta'',\Aa,\Delta' \succeq \Delta. 
    \wfres{\Delta''}{\Aa}{\alpha}{\Delta'} \impp
    \logrel{\Delta''}{M\conc \Aa=N\conc \Aa}{\tau}\\
  \end{eqnarray*}
\[
\begin{array}{c}
\infer{\logrel{\Delta}{\cdot = \cdot}{\cdot}}{}
\quad
\infer{\logrel{\Delta}{\theta,M/x=\sigma,N/x}{\Theta,\bind{x}{\tau}}}
{\logrel{\Delta}{M=N}{\tau} & \logrel{\Delta}{\theta=\sigma}{\Theta} }
\quad
\infer{\logrel{\Delta}{\theta,\Ab/\Aa=\sigma,\Ab/\Aa}{\Theta\#\bind{\Aa}{\alpha}}}
{\wfres{\Delta}{\Ab}{\alpha}{\Delta'} & \logrel{\Delta'}{\theta=\sigma}{\Theta} }
\end{array}
\]
\caption{Logical relation for objects and substitutions}\labelFig{logrel}
\end{figure*}


\begin{lem}
  [Weakening] If $\logrel{\Delta}{M=N}{\tau}$ and $\Delta' \succeq
  \Delta$ then $\logrel{\Delta'}{M=N}{\tau}$.
\end{lem}
\begin{proof}
  By induction on $\tau$.  The only new case is for name-abstraction
  types $\abs{\alpha}{\tau}$.  Suppose
  $\logrel{\Delta}{M=N}{\abs{\alpha}{\tau}}$ and $\Delta' \succeq
  \Delta$.  Let $\Delta'',\Delta''',\Aa$ be given with
  $\wfres{\Delta'''}{\Aa}{\alpha}{\Delta''}$ and $\Delta'' \succeq
  \Delta'$.  Then by transitivity we have $\Delta'' \succeq \Delta$ so
  by definition of the logical relation, $\logrel{\Delta'''}{M\conc
    \Aa = N \conc \Aa}{\tau}$.  Thus, we conclude that
  $\logrel{\Delta'}{M = N}{\abs{\alpha}{\tau}}$ by the definition of
  the logical relation.
\end{proof}

\begin{lem}
  [Symmetry] If $\logrel{\Delta}{M=N}{\tau}$ then $\logrel{\Delta}{N =
    M}{\tau}$.
\end{lem}
\begin{proof}
  The proof is by induction on types; we show the case for
  $\abs{\alpha}{\tau}$.  Assume
  $\logrel{\Delta}{M=N}{\abs{\alpha}{\tau}}$, and let
  $\Delta'',\Aa,\Delta'$ be given with
  $\wfres{\Delta''}{\Aa}{\alpha}{\Delta'}$ and $\Delta' \succeq
  \Delta$. Then by definition we have $\logrel{\Delta''}{M \conc \Aa=N
    \conc \Aa}{\tau}$ and by induction we have $\logrel{\Delta''}{N
    \conc \Aa=M \conc \Aa}{\tau}$ so we may conclude that
  $\logrel{\Delta}{N =M}{\tau}$.
\end{proof}

\begin{lem}
  [Transitivity] If $\logrel{\Delta}{M=N}{\tau}$ and
  $\logrel{\Delta}{N=O}{\tau}$ then $\logrel{\Delta}{M=O}{\tau}$.
\end{lem}
\begin{proof}
  The proof is by induction on types; we show the case for
  $\abs{\alpha}{\tau}$.  Suppose
  $\logrel{\Delta}{M=N}{\abs{\alpha}{\tau}}$ and
  $\logrel{\Delta}{N=O}{\abs{\alpha}{\tau}}$, and let
  $\Delta'',\Aa,\Delta'$ be given with
  $\wfres{\Delta''}{\Aa}{\alpha}{\Delta'}$ and $\Delta' \succeq
  \Delta$.  Then by definition we have both $\logrel{\Delta''}{M \conc
    \Aa=N \conc \Aa}{\tau}$ and $\logrel{\Delta''}{N \conc \Aa=O \conc
    \Aa}{\tau}$ and by induction we have $\logrel{\Delta''}{M \conc
    \Aa=O \conc \Aa}{\tau}$, so we may conclude that
  $\logrel{\Delta}{M = O}{\tau}$.
\end{proof}

\begin{lem}[Closure under head expansion]\labelLem{closure-whr}
  If $M \whr M'$ and $\logrel{\Delta}{M'=N}{\tau}$ then
  $\logrel{\Delta}{M=N}{\tau}$.
\end{lem}
\begin{proof}
  The proof is by induction on types; we show the case for
  $\abs{\alpha}{\tau}$.  Suppose $M \whr M'$ and
  $\logrel{\Delta}{M'=N}{\abs{\alpha}{\tau}}$.  Let $\Delta'',\Aa,\Delta'$ be given
  with $\wfres{\Delta''}{\Aa}{\alpha}{\Delta'}$ and $\Delta' \succeq
  \Delta$.  Then $\logrel{\Delta''}{M' \conc \Aa = N \conc \Aa}{\tau}$
  by definition of the logical relation.  Moreover, we have that $M
  \whr M'$ implies $M \conc \Aa \whr M' \conc \Aa$.  So, by induction
  we know that $\logrel{\Delta''}{M \conc \Aa = N \conc \Aa}{\tau}$,
  and we may conclude $\logrel{\Delta}{M=N}{\abs{\alpha}{\tau}}$.
\end{proof}

\begin{lem}[Identity substitution]\labelLem{identity-logrel}
  For any $\Gamma$ we have $\logrel{\Gamma^-}{\id_\Gamma =
    \id_\Gamma}{\Gamma^-}$.
\end{lem}
\begin{proof}
  Induction on the structure of $\Gamma$.  The base case and variable
  case are standard.  Suppose $\Gamma = \Gamma_0\#\Aa{:}\alpha$.
  Then by induction, $\logrel{\Gamma_0^-}{\id_{\Gamma_0} =
    \id_{\Gamma_0}}{\Gamma_0^-}$.  By weakening, we know that
  $\logrel{\Gamma_0^-\#\Aa{:}\alpha}{\id_{\Gamma_0} =
    \id_{\Gamma_0}}{\Gamma_0^-}$ holds.  Moreover,
  $\wfres{\Gamma_0^-\#\Aa{:}\alpha }{\Aa}{\alpha}{\Gamma_0^-}$ is
  derivable.  Hence, we may conclude:
\[
\infer{\logrel{\Gamma_0^-\#\Aa{:}\alpha}{\id_{\Gamma_0},\Aa/\Aa =
    \id_{\Gamma_0},\Aa/\Aa}{\Gamma_0^-\#\Aa{:}\alpha}}{
\hyp{\wfres{\Gamma_0^-\#\Aa{:}\alpha }{\Aa}{\alpha}{\Gamma_0^-}}
&
\logrel{\Gamma_0^-\#\Aa{:}\alpha}{\id_{\Gamma_0} =
    \id_{\Gamma_0}}{\Gamma_0^-}}
\]
This concludes the proof.
\end{proof}

We now state the main properties relating definitional and algorithmic
equality and the logical relation.  
\begin{thm}[Logical implies algorithmic]\labelThm{log-imp-alg}\hfill
  \begin{enumerate}[\em(1)]
  \item 
    If $\logrel{\Delta}{M=N}{\tau}$ then $\algeq{\Delta}{M}{N}{\tau}$.
    \item If $\streq{\Delta}{M}{N}{\tau}$ then
    $\logrel{\Delta}{M=N}{\tau}$.
  \end{enumerate}
\end{thm}
\proof
  By simultaneous induction on $\tau$.  The new cases are those for
  $\tau = \abs{\alpha}{\tau_0}$.
  \begin{enumerate}[(1)]
  \item  Suppose
    $\logrel{\Delta}{M=N}{\abs{\alpha}{\tau}}$.  Then we wish to show
    that $\algeq{\Delta}{M}{N}{\abs{\alpha}{\tau}}$.  Choose a fresh
    name $\Aa$ not present in $\Delta$.  Then we can immediately derive
    $\wfres{\Delta\#\Aa{:}\alpha}{\Aa}{\alpha}{\Delta}$, and obviously
    $\Delta \succeq \Delta$, so by definition of the logical relation,
    $\logrel{\Delta\#\Aa{:}\alpha}{M \conc \Aa=N \conc \Aa}{\tau}$.  By
    induction, we have $\algeq{\Delta\#\Aa{:}\alpha}{M \conc \Aa}{N
      \conc \Aa}{\tau}$, so we may conclude:
    \[
    \infer{\algeq{\Delta}{M}{N}{\abs{\alpha}{\tau}}}{
      \hyp{\wfres{\Delta\#\Aa{:}\alpha}{\Aa}{\alpha}{\Delta}}
      & \algeq{\Delta\#\Aa{:}\alpha}{M \conc \Aa}{N\conc \Aa}{\tau}}
    \]
  \item  Suppose $\streq{\Delta}{M}{N}{\abs{\alpha}{\tau}}$.  Let
    $\Delta',\Aa,\Delta''$ be given with
    $\wfres{\Delta''}{\Aa}{\alpha}{\Delta'}$ and $\Delta' \succeq
    \Delta$.  Then we may derive:
    \[
    \infer{\streq{\Delta''}{M \conc \Aa}{N \conc \Aa}{\tau}}{
      \wfres{\Delta''}{\Aa}{\alpha}{\Delta'} 
      &
      \infer[W]{\streq{\Delta'}{M}{N}{\abs{\alpha}{\tau}}}{
        \streq{\Delta}{M}{N}{\abs{\alpha}{\tau}}
      }
    }
    \]
    where the step labeled $W$ is by weakening using $\Delta \preceq
    \Delta'$.  Hence, the induction hypothesis applies and we have
    $\logrel{\Delta''}{M \conc \Aa = N \conc \Aa}{\tau}$, so we may
    conclude by definition that $\logrel{\Delta}{M = N}{\abs{\alpha}{\tau}}$.
  \end{enumerate}
This completes the proof. \qed

\begin{thm}[Definitional implies logical]\labelThm{def-imp-log}
  If $\wf{\Gamma}{M=N}{A}$ and
  $\logrel{\Delta}{\theta=\sigma}{\Gamma^-}$ then
  $\logrel{\Delta}{M[\theta]=N[\sigma]}{A^-}$.
\end{thm}
\proof
  By induction on the definitional equality derivation.  We show new
  cases involving new definitional equality rules.
  \begin{iteMize}{$\bullet$}
  \item If the derivation is of the form:
    \[
    \infer[\mathsf{eq\_nm}]{\wf{\Gamma}{\Aa=\Aa}{\alpha}}{\Aa{:}\alpha
      \in \Gamma}
    \]
    then it is immediate that $\algeq{\Gamma^-}{\Aa}{\Aa}{\alpha}$ and
    hence $\logrel{\Gamma^-}{\Aa=\Aa}{\alpha}$.
  \item If the derivation is of the form:
    \[
    \infer[\mathsf{eq\_abs}]{\wf{\Gamma}{\abs{\Aa{:}\alpha}{M}=\abs{\Aa{:}\alpha}{N}}{\new
        \Aa{:}\alpha.A}}{ \wf{\Gamma\#\Aa{:}\alpha}{M=N}{A}}
    \]
    then we wish to show that $\logrel{\Delta''}{
      (\abs{\Aa{:}\alpha}{M}) [\theta]=(\abs{\Aa{:}\alpha}{N})
      [\sigma]}{\abs{\alpha}{A^-}}$.  To prove this, suppose
    $\Delta',\Delta'',\Ab$ are given with
    $\wfres{\Delta''}{\Ab}{\alpha}{\Delta'}$ and $\Delta' \succeq
    \Delta$.  Using logical relation weakening, we have that
    $\logrel{\Delta'}{\theta=\sigma}{\Gamma^-}$.  So we may derive
    \[
    \infer{\logrel{\Delta''}{\theta,\Ab/\Aa=\sigma,\Ab/\Aa}{\Gamma^-\#\Aa{:}\alpha}}
    {\wfres{\Delta''}{\Ab}{\alpha}{\Delta'} &
      \logrel{\Delta'}{\theta=\sigma}{\Gamma^-} }
    \]
    So by induction, we have $\logrel{\Delta''}{M[\theta,\Ab/\Aa] =
      N[\sigma,\Ab/\Aa]}{A^-}$.  Moreover,
    \[
    (\abs{\Aa{:}\alpha}{M})[\theta] \conc \Ab =
    (\abs{\Aa{:}\alpha}{M[\theta]}) \conc \Ab\whr M[\theta][\Ab/\Aa] =
    M[\theta,\Ab/\Aa]\;.\]
    Similarly,
    \[(\abs{\Aa{:}\alpha}{N})[\sigma] \conc \Ab =
    (\abs{\Aa{:}\alpha}{N[\sigma]}) \conc \Ab\whr N[\sigma][\Ab/\Aa] =
    N[\sigma,\Ab/\Aa]\;.\]
    Hence, using \refLem{closure-whr}, we can conclude that
    $\logrel{\Delta''}{ (\abs{\Aa{:}\alpha}{M}) [\theta] \conc
      \Ab=(\abs{\Aa{:}\alpha}{N}) [\sigma] \conc \Ab}{A^-}$.
    Moreover, since $\Delta'',\Delta',\Ab$ were arbitrary, we have
    that $\logrel{\Delta''}{ (\abs{\Aa{:}\alpha}{M})
      [\theta]=(\abs{\Aa{:}\alpha}{N}) [\sigma]}{\abs{\alpha}{A^-}}$,
    as desired.
  \item If the derivation is of the form:
    \[
    \infer[\mathsf{eq\_conc}]{\wf{\Gamma}{M\conc \Ab = N \conc
        \Ab}{A[\Ab/\Aa]}}{ \wfres{\Gamma}{\Ab}{\alpha}{\Gamma_0} &
      \wf{\Gamma_0}{M=N}{\new \Aa{:}\alpha.A}}
    \]
    then we wish to show that $\logrel{\Delta}{(M \conc \Ab)[\theta] =
      (N \conc \Ab)[\sigma]}{A^-}$ (noting that $A[\Ab/\Aa]^- = A^-$).
    By \refLem{subst-restrict}, we know that $\theta(\Ab) =
    \sigma(\Ab)$ and there must exist $\Delta_0$ such that
    $\wfres{\Delta}{\theta(\Ab)}{\alpha}{\Delta_0}$ and
    $\logrel{\Delta_0}{\theta-\Ab=\sigma-\Ab}{\Gamma_0^-}$.  Moreover,
    by induction we have that
    $\logrel{\Delta_0}{M[\theta-\Ab]=N[\sigma-\Ab]}{\abs{\alpha}{A^-}}$.
    Observe that
    $\wfres{\Delta_0\#\theta(\Ab){:}\alpha}{\theta(\Ab)}{\alpha}{\Delta_0}$
    is immediately derivable, and that $\Delta_0 \succeq \Delta_0$
    trivially holds.  Thus, by definition we have
    $\logrel{\Delta_0\#\theta(\Ab){:}\alpha}{M[\theta-\Ab] \conc
      \theta(\Ab) = N[\sigma-\Ab] \conc \theta(\Ab)}{A^-}$.  To
    conclude, we observe that $M[\theta-\Ab] \conc \theta(\Ab) = (M
    \conc \Ab) [\theta]$ and $N[\sigma-\Ab] \conc \theta(\Ab) =
    (N\conc \Ab)[\sigma]$ since $\theta(\Ab) = \sigma(\Ab)$, and in
    addition $\Delta_0\#\theta(\Ab){:}\alpha \preceq \Delta$ so by
    weakening we have $\logrel{\Delta}{(M \conc \Ab) [\theta] =
      (N\conc \Ab)[\sigma]}{A^-}$, as desired.
  \item If the derivation is of the form:
    \[
    \infer[\mathsf{eq\_nm\_beta}]{\wf{\Gamma}{(\abs{\Aa{:}\alpha}M)\conc
        \Ab = N[\Ab/\Aa]}{A[\Ab/\Aa]}}{
      \wfres{\Gamma}{\Ab}{\alpha}{\Gamma_0} &
      \wf{\Gamma_0\#\Aa{{:}}\alpha}{M=N}{A}}
    \]
    then we must show that
    $\logrel{\Delta}{((\abs{\Aa{:}\alpha}M)\conc \Ab)[\theta] =
      (N[\Ab/\Aa])[\sigma]}{A^-}$, again noting $A^- = A[\Ab/\Aa]^-$.
    Again using \refLem{subst-restrict}, we know that $\theta(\Ab) =
    \sigma(\Ab)$ and there must exist $\Delta_0$ such that
    $\wfres{\Delta}{\theta(\Ab)}{\alpha}{\Delta_0}$ and
    $\logrel{\Delta_0}{\theta-\Ab=\sigma-\Ab}{\Gamma_0^-}$.
    Moreover, we can derive
    \[
    \infer{\logrel{\Delta_0\#\theta(\Ab){:}\alpha}{(\theta-\Ab),\theta(\Ab)/\Aa
        = (\sigma-\Ab),\theta(\Ab)/\Aa}{\Gamma_0^-\#\Aa{:}\alpha}}{
      \infer{\wfres{\Delta_0\#\theta(\Ab){:}\alpha}{\theta(\Ab)}{\alpha}{\Delta_0}}{}
      & 
      \logrel{\Delta_0}{\theta-\Ab = \sigma-\Ab}{\Gamma_0^-}
    }
    \]
    and so, by induction, we have
    $\logrel{\Delta_0\#\theta(\Ab){:}\alpha}{M[\theta-\Ab,\theta(\Ab)/\Aa]
      = N[\sigma-\Ab,\theta(\Ab)/\Aa]}{A^-}$.  
    Now we observe that:
    \begin{eqnarray*}
      ((\abs{\Aa{:}\alpha}M)\conc \Ab)[\theta] 
      &=& (\abs{\Aa{:}\alpha}M)[\theta]\conc \Ab[\theta] \\
&=& (\abs{\Aa{:}\alpha}M[\theta])\conc \theta(\Ab) \\
      &\whr& M[\theta][\theta(\Ab)/\Aa]\\
&=& M[\theta-\Ab][\theta(\Ab)/\Aa]= M[\theta-\Ab,\theta(\Ab)/\Aa]
    \end{eqnarray*}
    and 
    \begin{eqnarray*}
      N[\sigma-\Ab,\theta(\Ab)/\Aa] &=& N[\sigma-\Ab][\theta(\Ab)/\Aa] \\
&=&
    N[\sigma][\sigma(\Ab)/\Aa] \\
&=& N[\Ab/\Aa][\sigma]\;.
    \end{eqnarray*}
    Hence, by \refLem{closure-whr} and weakening $
    \Delta_0\#\theta(\Ab){:}\alpha \preceq \Delta$ we can conclude
    $\logrel{\Delta}{((\abs{\Aa{:}\alpha}M)\conc \Ab)[\theta] 
      = N[\Ab/\Aa][\sigma]}{A^-}$ as desired.

  \item If the derivation is of the form:
    \[
    \infer[\mathsf{eq\_nm\_eta}]
    {\wf{\Gamma}{M=N}{\new \Aa{{:}}\alpha.A}}
    {\wf{\Gamma\#\Aa{:}\alpha}{M\conc \Aa = N\conc \Aa}{A}}
    \] 
    then we wish to show that
    $\logrel{\Delta}{M[\theta]=N[\sigma]}{\abs{\alpha}{A^-}}$.  To
    prove this, let $\Delta',\Delta'',\Ab$ be given such that
    $\wfres{\Delta''}{\Ab}{\alpha}{\Delta'}$ and $\Delta' \succeq
    \Delta$.  We may then derive:
    \[\infer{\logrel{\Delta''}{\theta,\Ab/\Aa = \sigma,\Ab/\Aa}{\Gamma^-\#\Aa{:}\alpha}}
    { \wfres{\Delta''}{\Ab}{\alpha}{\Delta'} &
      \infer[W]{\logrel{\Delta'}{\theta=\sigma}{\Gamma^-}}
      {
        \logrel{\Delta}{\theta=\sigma}{\Gamma^-}
      }
    }
    \]
    where the step labeled $W$ is by logical relation weakening.  So,
    by induction, we obtain $\logrel{\Delta''}{(M\conc
      \Aa)[\theta,\Ab/\Aa] = (N \conc \Aa)[\sigma,\Ab/\Aa]}{A^-}$.
    Moreover, we calculate $(M \conc \Aa)[\theta,\Ab/\Aa] =
    M[\theta,\Ab/\Aa] \conc \Ab = M[\theta] \conc \Ab$ since $\Aa$
    must not appear in $M$.  Similarly, $(N \conc \Aa)[\sigma,\Ab/\Aa]
    = N[\sigma]\conc\Ab$.  We thus have $\logrel{\Delta''}{M[\theta]
      \conc \Ab = N[\sigma]\conc\Ab}{A^-}$, as desired to show
    $\logrel{\Delta}{M[\theta]=N[\sigma]}{\abs{\alpha}{A^-}}$.
  \end{iteMize}
This completes the proof.\qed
 
\begin{thm}[Completeness]
  If $\wf{\Gamma}{M=N}{A}$ then $\algeq{\Gamma^-}{M}{N}{A^-}$.
\end{thm}
\begin{proof}
  Immediate, combining \refLem{identity-logrel}, \refThm{def-imp-log},
  and \refThm{log-imp-alg}.
\end{proof}

\subsection{Decidability,  canonical forms and conservativity}
Once we have established that algorithmic equivalence is sound and
complete for well-formed terms, we can also extend the algorithmic
typechecking rules in Harper and Pfenning's system to handle
name-abstractions and verify that all judgments are decidable:
\begin{thm}[Decidability]
  All judgments of \DNTT are decidable.
\end{thm}

\begin{figure*}[tb]
  \centering
  \[
  \begin{array}{c}
    \infer{\atomic(a)}{}
    \quad
    \infer{\atomic(A~M)}{\atomic(A)}
    \bigskip\\
    \infer{\wfatm{\Gamma}{c}{c}{A}}{c:A\in\Sigma}
    \quad
    \infer{\wfatm{\Gamma}{x}{x}{A}}{x:A\in\Gamma}
    \quad
    \infer{\wfatm{\Gamma}{M~N}{M'~N'}{B[N'/x]}}{\wfatm{\Gamma}{M}{M'}{\PI{x}{A}B} & \wfcan{\Gamma}{N}{N'}{A}}
    \smallskip\\
    \infer{\wfatm{\Gamma}{\Aa}{\Aa}{\alpha}}{\Aa:\alpha\in\Gamma}
    \quad
    \infer{\wfatm{\Gamma}{M\conc \Aa}{N\conc \Aa}{B}}{\wfres{\Gamma}{\Aa}{\alpha}{\Gamma'} & \wfatm{\Gamma'}{M}{N}{\NEW{\Aa}{\alpha}{B}}}
    \smallskip\\
    \infer{\wfcan{\Gamma}{M}{\lambda x{:}A.N}{\PI{x}{A}{B}}}{\wfcan{\Gamma,x:A}{M~x}{N}{B}}
    \quad
    \infer{\wfcan{\Gamma}{M}{\abs{\Aa{:}\alpha}{N}}{\NEW{\Aa}{\alpha}{A}}}{\wfcan{\Gamma\#\bind{\Aa}{\alpha}}{M\conc \Aa}{N}{B}}
    \smallskip\\
    \infer{\wfcan{\Gamma}{M}{N}{A}}{\wfatm{\Gamma}{M}{N}{A} & \atomic(A)}
\quad
    \infer{\wfcan{\Gamma}{M}{N}{A}}{M \whr M' & \wfcan{\Gamma}{M'}{N}{A} & \atomic(A)}
\bigskip\\
\infer{\wfatm{\Gamma}{a}{a}{K}}{a:K \in \Sigma}
\quad
\infer{\wfatm{\Gamma}{\alpha}{\alpha}{K}}{\alpha{:}\nametype \in \Sigma}
\quad
\infer{\wfatm{\Gamma}{A~M}{A'~M'}{K[M'/x]}}{\wfatm{\Gamma}{A}{A'}{\Pi x{:}A.K} & \wfcan{\Gamma}{M}{M'}{A}}
\smallskip\\
\infer{\wfcan{\Gamma}{A}{A'}{\type}}{\wfatm{\Gamma}{A}{A'}{\type}}
\quad
\infer{\wfcan{\Gamma}{\PI{x}{A}{B}}{\PI{x}{A'}{B'}}{\type}}{\wfcan{\Gamma}{A}{A'}{\type} & \wfcan{\Gamma,x{:}A'}{B}{B'}{\type}}
\quad
\infer{\wfcan{\Gamma}{\NEW{\Aa}{\alpha}{B}}{\NEW{\Aa}{\alpha}{B'}}{\type}}{\wfcan{\Gamma\#\Aa{:}\alpha}{B}{B'}{\type}}
\bigskip\\
\infer{\wfcan{\Gamma}{\type}{\type}{\kind}}{}
\quad
\infer{\wfcan{\Gamma}{\PI{x}{A}{K}}{\PI{x}{A'}{K'}}{\kind}}{\wfcan{\Gamma}{A}{A'}{\type} & \wfcan{\Gamma,x{:}A'}{K}{K'}{\kind}}
\end{array}
\]
  \caption{Canonicalization}
  \label{fig:canonicalization}
\end{figure*}

We say that a \DNTT expression is in \emph{canonical form} if it is
$\beta$-normal and cannot be $\eta$-expanded without introducing a
$\beta$-redex.  Canonical forms of \DNTT are similar to those for LF,
but can include name-abstractions and concretions.  The following
grammar describes the syntax of canonical and atomic forms:
\begin{eqnarray*}
  M_c &::=& \lambda x{:}A_c.M_c \mid \abs{a{:}\alpha}{M_c} \mid M_a\\
M_a &::=& c \mid \Aa \mid M_a~M_c \mid M_a\conc \Aa\\
A_c &::=& a \mid \alpha \mid A_c~M_c \mid \PI{x}{A_c}{B_c} \mid
\NEW{\Aa}{\alpha}{A_c}\\
K_c &::=& \type \mid \PI{x}{A_c}{K_c}
\end{eqnarray*}
Note, however, that not all terms matching the above grammar are in
canonical or atomic form; further typing constraints are needed to
ensure full $\eta$-expansion.  We give an inference rule system for
canonicalizing object terms, which also implicitly gives the typing
constraints that canonical forms must satisfy, in
\refFig{canonicalization}.  In particular, the $atomic(-)$ predicate
is used to restrict weak head normalization and ensure only atomic
forms whose type is an atomic type $A~M_1~\cdots~M_n$ can be
considered canonical.

We will show:
\begin{thm}[Canonical forms]\labelThm{canonical}\hfill
  Assume that all the types and kinds in $\Gamma$, $\Sigma$ and $A$
  are in canonical form.  Then:
  \begin{enumerate}[\em(1)]\item 
    If $\wf{\Gamma}{M}{A}$ then there exists a canonical $P$ such that
    $\wfcan{\Gamma}{M}{P}{A}$ and $\wf{\Gamma}{M=P}{A}$.
  \item If $P'$ also satisfies $\wfcan{\Gamma}{M}{P'}{A}$, then
    $P=P'$.
  \item If $\wf{\Gamma}{M=N}{A}$ holds, then their canonical forms are
    equal.
  \end{enumerate}
\end{thm}

To show the canonicalization theorem, we first show the stronger
property:
\begin{lem}[Algorithmically equivalent terms have common canonical
  forms]\hfill
Assume that all types and kinds in $\Sigma$, $\Gamma$, $A$ and $B$ are
in canonical form.  Then:
  \begin{enumerate}[\em(1)]
  \item If $\wf{\Gamma}{M}{A}$ and $\wf{\Gamma}{N}{B}$ and
    $\streq{\Gamma^-}{M}{N}{\tau}$ then $\wf{\Gamma}{A=B}{\type}$ and
    $A^-=B^-=\tau$ and there exists $P$ such that
    $\wfatm{\Gamma}{M}{P}{A}$ and $\wfatm{\Gamma}{N}{P}{A}$.
  \item 
    If $\wf{\Gamma}{M}{A}$ and $\wf{\Gamma}{N}{A}$ and
    $\algeq{\Gamma^-}{M}{N}{A^-}$ then there exists $P$ such that
    $\wfcan{\Gamma}{M}{P}{A}$ and $\wfcan{\Gamma}{N}{P}{A}$.
 \end{enumerate}
\end{lem}
\begin{proof}
  By structural induction on the algorithmic derivations, using
  inversion and injectivity of products as appropriate.  For the
  ordinary cases, we need the assumption that $\Sigma,\Gamma,A,B$ are
  already canonical in order to ensure that type tags in $M,N$ are
  compatible.  We show the cases specific to \DNTT:
\begin{iteMize}{$\bullet$}
\item If the derivation is of the form
\[
\infer{\streq{\Gamma^-}{\Aa}{\Aa}{\alpha}}{\bind{\Aa}{\alpha}\in\Gamma^-}
\]
then we must have that $M = \Aa = N$ and $A = \alpha = B$ and
$\Aa{:}\alpha \in \Gamma$, so we can conclude that
$\wf{\Gamma}{\alpha=\alpha}{\type}$ and derive
\[
\infer{\wfatm{\Gamma}{\Aa}{\Aa}{\alpha}}{\Aa{:}\alpha \in \Gamma}
\quad
\infer{\wfatm{\Gamma}{\Aa}{\Aa}{\alpha}}{\Aa{:}\alpha \in \Gamma}
\]

\item If the derivation is of the form
\[
\infer{\streq{\Gamma^-}{M\conc \Aa}{N\conc \Aa}{\tau}}{\wfres{\Gamma^-}{\Aa}{\alpha}{\Delta'} & \streq{\Delta'}{M}{N}{\abs{\alpha}{\tau}}}
\]
By inversion we have $\wfres{\Gamma}{\Aa}{\alpha_1}{\Gamma_1}$ and
$\wf{\Gamma_1}{M}{\NEW{\Aa}{\alpha_1}{A_1}}$.  Similarly, we have
$\wfres{\Gamma}{\Aa}{\alpha_2}{\Gamma_2}$ and
$\wf{\Gamma_2}{N}{\NEW{\Aa}{\alpha_2}{A_2}}$.  Moreover we must have
$\alpha_1 = \alpha_2$ and $\Gamma_1 = \Gamma_2$; also, we must have
$\Gamma_0^- = \Delta'$.  So, the induction hypothesis applies and we
know that $\wf{\Gamma_1}{ \NEW{\Aa}{\alpha_1}{A_1} =
  \NEW{\Aa}{\alpha_2}{A_2}}{\type}$ and $(\NEW{\Aa}{\alpha_1}{A_1})^-
= \abs{\alpha}{\tau} = (\NEW{\Aa}{\alpha_2}{A_2})^-$, which implies
that $\alpha_1 = \alpha = \alpha_2$ and $A_1^- = \tau = A_2^-$.
In fact, since $A_1$ and $A_2$ are in canonical form already, we must
have $A_1 = A_2$.
Furthermore, by induction we also have
$\wfatm{\Gamma}{M}{P}{\NEW{\Aa}{\alpha}{A_1}}$ and
$\wfatm{\Gamma}{N}{P}{\NEW{\Aa}{\alpha}{A_1}}$.  To conclude, we may
derive:
\[
\infer{\wfatm{\Gamma}{M \conc \Aa}{P \conc \Aa}{A_1}}{\wfres{\Gamma}{\Aa}{\alpha}{\Gamma_1} & \wfatm{\Gamma_1}{M}{P}{\NEW{\Aa}{\alpha}{A_1}}}
\]\[
\infer{\wfatm{\Gamma}{N \conc \Aa}{P\conc \Aa}{A_1}}{\wfres{\Gamma}{\Aa}{\alpha}{\Gamma_1} & \wfatm{\Gamma_1}{N}{P}{\NEW{\Aa}{\alpha}{A}_1}}
\]

\item If the derivation is of the form
\[
\infer{\algeq{\Gamma^-}{M}{N}{A^-}}{\algeq{\Gamma^-\#\bind{\Aa}{\alpha}}{M\conc \Aa}{N \conc \Aa}{\tau}}
\]
then we must have that $A^- = \abs{\alpha}{\tau}$ for some $\alpha$
and $\tau$ and so $A$ must be of the form $\NEW{\Aa}{\alpha}{B}$ where
$B^- = \tau$.  Thus, we have derivation $\algeq{\Gamma^-\#\bind{\Aa}{\alpha}}{M\conc \Aa}{N \conc \Aa}{B^-}$.  Moreover, we can derive $\wf{\Gamma\#\bind{\Aa}{\alpha}}{M\conc \Aa}{B}$ and $\wf{\Gamma\#\bind{\Aa}{\alpha}}{N\conc \Aa}{B}$.  So by induction we have derivations $\wfcan{\Gamma\#\bind{\Aa}{\alpha}}{M\conc \Aa}{P}{B}$ and $\wfcan{\Gamma\#\bind{\Aa}{\alpha}}{N\conc \Aa}{P}{B}$, so we can conclude by deriving:
\[\infer{\wfcan{\Gamma}{M}{\abs{\Aa{:}\alpha}{P}}{\NEW{\Aa}{\alpha}{B}}}{\wfcan{\Gamma\#\bind{\Aa}{\alpha}}{M\conc \Aa}{P}{B}}
\qquad 
\infer{\wfcan{\Gamma}{N}{\abs{\Aa{:}\alpha}{P}}{\NEW{\Aa}{\alpha}{B}}}{\wfcan{\Gamma\#\bind{\Aa}{\alpha}}{N\conc \Aa}{P}{B}}
\]
\end{iteMize}
\end{proof}

We also can easily show that canonicalization is sound with respect to
definitional equivalence:

\begin{lem}[Soundness of canonicalization]\hfill
  \begin{enumerate}[\em(1)]\item 
    If $\wfatm{\Gamma}{M}{P}{A}$ then $\wf{\Gamma}{M=P}{A}$.
  \item
    If $\wfcan{\Gamma}{M}{P}{A}$ then $\wf{\Gamma}{M=P}{A}$.
  \end{enumerate}
\end{lem}

We also need to show that the canonicalization judgment is deterministic:
\begin{lem}[Determinism of canonicalization]\hfill
  \begin{enumerate}[\em(1)]\item 
    If $\wfatm{\Gamma}{M}{P}{A}$ and $\wfatm{\Gamma}{M}{P'}{A'}$ then
    $P=P'$ and $A = A'$.
  \item If $\wfcan{\Gamma}{M}{P}{A}$ and $\wfcan{\Gamma}{M}{P'}{A}$
    then $P=P'$.
    \end{enumerate}
\end{lem}
\begin{proof}
  By induction on derivations and inversion.
\end{proof}

The above lemmas imply the first and second parts of the
Canonicalization Theorem.  The third part follows by inspection of the
rules for canonicalization, since if $A$ and $\Gamma$ are already in
canonical form then any types that are copied into the result of
canonicalization will also be canonical.

Moreover, we can use the canonicalization rules for types and kinds
shown in \refFig{canonicalization} to canonicalize $\Sigma$,
$\Gamma$ and $A$, so we have the following stronger result:
\begin{thm}
  If $\Sigma$ and $\Gamma$ are in canonical form and
  $\wf{\Gamma}{M}{A}$ then there exist unique canonical $A'$ and $M'$
  such that $\wf{\Gamma}{A = A'}{\type}$ and $\wf{\Gamma}{M=M'}{A'}$.
\end{thm}

Finally, the canonical forms theorem implies \DNTT is a conservative
extension of LF in the sense that it introduces no new derivable LF judgments.
\begin{cor}[Conservativity] 
  If $\judge{\Gamma}{\JJ}$ is an LF judgment over a valid LF signature
  $\Sigma$ and is derivable in \DNTT, then $\judge{\Gamma}{\JJ}$ is
  derivable in LF.
\end{cor}

\section{Adequacy}\labelSec{adequacy}

It is a significant concern whether a given signature correctly
represents an object language we have in mind.  This property is often
referred to as \emph{adequacy} in an LF
settings~\cite{harper05tocl,cheney11jar}.  As in LF, adequacy in \DNTT
relies upon the existence of (unique) canonical forms.

In this section, we sketch an adequacy argument for a typical object
language, the untyped lambda-calculus equipped with an inequality
predicate (as shown in the introduction).  

Recall the signature given in \refFig{lambda}.  The canonical forms of
expressions of type $\expty$ in \DNTT are generated by the grammar:
\[M_0,N_0 ::= \var~\Ax\mid \app~M_0~N_0\mid
\lam~\abs{\Ax{:}\expty}{M_0}\] 
The encoding is defined on object-language terms as follows:
\[
\enc{\Ax} = \var~\Ax\qquad
\enc{t~u} = \app~\enc{t}~\enc{u}\qquad
\enc{\lambda \Ax.t} = \lam~\abs{\Ax}{\enc{t}}
\]
The main result concerning the correctness of the encoding is:
\begin{thm}[Adequacy of encoding]
  The encoding function $\enc{-}$ is injective and maps object
  language terms $t$ (having free variables $\Ax_1,\ldots,\Ax_n$) onto the
  set of canonical forms of type $\expty$ (in context
  $\Ax_1{:}\varty\#\ldots\#\Ax_n{:}\varty$).  Moreover, the encoding function commutes
  with renaming, that is, $\enc{t[\Ax/\Ay]} = \enc{t}[\Ax/\Ay]$.
\end{thm}

Furthermore, we can reason by inversion on canonical forms to
establish that the alpha-inequality judgment holds precisely for terms
whose encodings are different modulo alpha-equivalence:
\begin{thm}[Adequacy of $\neqq$]
  Suppose we have object terms $t,u$ with free variables
  $\Ax_1,\ldots,\Ax_n$.  Then $t \not\equiv_\alpha u$ if and only if
  $\wf{\Ax_1{:}\varty\#\ldots\#\Ax_n{:}\varty}{\DD}{\neqq~\enc{t}~\enc{u}}$ is
  derivable for some (canonical) $\DD$.
\end{thm}
\begin{proof}
  The forward direction is straightforward.  The reverse direction is
  proved by induction on the canonical form of the proof term $\DD$.
  One key case is when $\DD$ is of the form
  $\neqqxx{v}{v}\conc\Ax_i\conc\Ax_j$.  In this case, we must have
  $\enc{t} = \var~\Ax_i$ and $\enc{u} = \var~\Ax_j$ for some $i \neq
  j$, since otherwise $\DD$ would be ill-formed.  Clearly, then $t$ must
  be $\Ax_i$ and $u$ must be $\Ax_j$ which are not
  $\alpha$-equivalent.

  Another key case is that for $\DD = \neqqxx{l}{l}~M_1~M_2~\DD' :
  \neqq~(\enc{t_1})~(\enc{t_2})$.  In this case, we know that
  $\enc{t_1} = \lam ~M_1$ and $\enc{t_2}=\lam~M_2$, so $t_1 = \lambda
  \Ax.t_1'$ and $t_2 = \lambda \Ax.t_2'$ for some $\Ax,t_1',t_2'$
  (without loss of generality we can assume the same name $\Ax$ is
  used for both and $\Ax$ is fresh for all other terms).  Hence $M_1 =
  \abs{\Ax}{\enc{t_1'}}$ and $M_2 = \abs{\Ax}{\enc{t_2'}}$ which means
  that the subderivation $\DD'$ must have type $\new
  \Ax.\neqq~((\abs{\Ax}{\enc{t_1'}})\conc
  \Ax)~((\abs{\Ax}{\enc{t_2'}})\conc \Ax)$.  By weakening the context
  to include name $\Ax:\varty$ and $\beta$-converting, we can see that
  $\DD'\conc \Ax$ must also have type $\neqq~(\enc{t_1'})~(\enc{t_2'})$.
  Moreover, $\DD' \conc \Ax$ must have a canonical form of this type,
  and so by induction we know that $t_1' \not\equiv_\alpha t_2'$.
  This implies $t_1 = \lambda \Ax.t_1' \not\equiv_\alpha \lambda
  \Ax.t_2' = t_2$.
\end{proof}

\section{Extensions and Examples}\labelSec{examples}

In previous work on a simple nominal type theory~\cite{cheney08lfmtp}
we discussed extensions such as name-comparison operations, lists,
datatypes involving name-binding, and recursion combinators for
defining functions over such datatypes.  These extensions were
motivated by a denotational interpretation of SNTT using nominal sets
(following~\cite{pitts03ic}). We will not develop a denotational
semantics of \DNTT here; however, the topos of nominal sets provides
all of the necessary structure to interpret dependent types, and it
seems clear that the extensions we consider can be justified using
Sch\"opp and Stark's semantics for a more general nominal type
theory~\cite{schoepp04csl,schoepp06phd} or using Pitts' approach to
recursion in a slightly different nominal type
theory~\cite{pitts10popl,pitts11jfp}.  

In this section we recapitulate and generalize extensions for
name-comparison, recursive function definitions and inductive
reasoning in \DNTT.  The computational extensions can easily be proved
type-sound but do not necessarily preserve the canonicalization or
decidability properties established earlier; we expect that these
extensions would be more relevant to intensional type theories where
only $\beta$-normalization results are needed.  We also discuss
applications of \DNTT as a framework for defining logics and for
encoding proof terms about languages with names and binding.

\paragraph{Name-comparison}
First, we consider a name comparison operation:
\[
  \begin{array}{l}
cond_\alpha : \Abs{\alpha}{\alpha} \to A \to (\alpha \to A) \to A\\
cond_\alpha ~(\abs{\Ax}{\Ax})~M~N \to_\beta M \qquad
cond_\alpha ~(\abs{\Ax}{\Ay})~M~N \to_\beta N~\Ay
\end{array}
\]
This takes a name-abstraction and two additional arguments
$M:A,N:\alpha\to A$.  If the abstraction is of the form
$\abs{\Ax}{\Ax}$, we return $M$, otherwise, if it is of the form
$\abs{\Ay}{\Ax}$ where $\Ax \neq \Ay$, we return $N~\Ax$.  Note that
it would make little sense to allow the type $A$ to depend on $\Ax$
since $\Ax$ may not ``escape'' in the first case.  

\paragraph{Recursion}

Now consider the standard nominal datatype encoding of the
lambda-calculus introduced in the introduction (\refFig{lambda}).
This datatype admits an obvious dependently-typed recursion principle:
\[
\begin{array}{rcl}
  \rec_{\expty}^T &:& (\Pi X{:}\varty. T~(\var~X)) \to\\
  && (\Pi M,N{:}\expty.  T~M \to T~N \to T~(\app~M~N)) \to\\
  &&(\Pi M{:}\Abs{\varty}{\expty}.  (\new \Aa{:}\varty. T~(M\conc \Aa)) \to  T~(\lam~M))\to \\
&&  \Pi M{:}\expty. T~M
\end{array}
\]
for any $T : \Pi x:\expty.\type$.  We also equip $\rec^T_{\expty}$ with the obvious rewriting rules for $\var$ and $\app$, along with
\[\rec_{\expty}^T ~ V~A~L~(\lam~F) \to_\beta L(\abs{\Aa{:}\varty}{\rec_{\expty}^T~V~A~L~(F\conc \Aa)})\]
(provided $\Aa \notin FV(V,A,L,F)$) for lambda-abstractions.

\subsection{Closure conversion}

\emph{Closure conversion} (see for example~\cite{appel92compiling}) is
an important transformation in functional language compilation.  A
function is closed if it refers only to its argument and locally
defined variables, not to variables whose scope began outside the
function.  Closure conversion translates an arbitrary expression to
one containing only closed functions.  There are many ways of doing
this, embodying different approaches to managing the environment.  We
consider a simplistic approach in which each function is translated to
a pair consisting of a closed function and an environment containing
all non-local variable values.  We define the translation of a term
$e$ that is well-formed in context $\Gamma$ and environment $env$ as
$C\SB{\Gamma \nd e}env$, where
\[\begin{array}{rcl}
  C\SB{\Gamma,x \nd x}env &=& \pi_1(env)\\ 
  C\SB{\Gamma,x \nd y}env &=& C\SB{\Gamma\nd y}\pi_2(env) \\ 
  C\SB{\Gamma \nd e_1~e_2}env &=& \mathbf{let}~z=C\SB{\Gamma \nd e_1}env\\
&&\mathbf{in}~(\pi_1(z))~\ab{C\SB{\Gamma \nd e_2}env,\pi_2(z)}\\
  C\SB{\Gamma \nd \lambda x.e}env &=& \ab{\lambda y.C\SB{\Gamma,x\nd e}y,env}
\end{array}\]
where $x \neq y$ in the second equation, $z \notin FV(\Gamma,e_1,e_2)$
in the third, and $y \not\in FV(\Gamma,x,e,e_0)$ in the fourth.  Note
that we include let-bindings here for convenience.  

\begin{figure*}[tb]
\[\begin{array}{lcl}
unit &:& exp.\\
pair &:& exp \to exp \to exp.\\
pi_1  &:& exp \to exp.\\
pi_2  &:& exp \to exp.\\
let  &:& exp \to (\Abs{id}{exp}) \to exp.\\
\\
cconv      &:&  list\ id \to exp \to exp \to exp \to \type.\\
cconv\_v1 &:&  cconv\ [G,X]\ (var\ X)\ Env\ (pi_1\ Env).\\
cconv\_v2 &:&  cconv\ [G,X]\ (var\ Y)\ Env\ E \gets neq\ X\ Y\\
&\gets& cconv\ G\ (var\ Y)\ (pi_2\ Env)\ E.\\
cconv\_a  &:&  cconv\ G\ (app\ E_1\ E_2)\ Env\ (let\ E_{11}\ (\ABS{\Az}{id}{app\ (pi_1(var(\Az)))}~ (pair\ E_{21}\ (pi_2(var(\Az))))))\\
           &\gets& cconv\ G\ E_1\ Env\ E_{11}\\
&\gets& cconv\ G\ E_2\ Env\ E_{21}.\\
cconv\_l  &:&  cconv\ G\ (lam\ F_1)\ Env\ (pair\ (lam\ F_2)\ Env)\\
&\gets& \new \Ax.\new \Ay.cconv\ [G,\Ax]\ (F_1@\Ax)\ (var\ \Ay)\ (F_2@\Ay).\\
\end{array}\]
\caption{Closure conversion translation}
\labelFig{closure-conversion}
\end{figure*}

\begin{exa}
  As a simple example, consider the closure-conversion of the
  $K$-combinator:
\[ \begin{array}{rcl}
    C\SB{\Gamma \nd \lambda x.\lambda y.x}env 
&=& \ab{\lambda x'.C\SB{\Gamma,x\nd \lambda y. x}x',env}\\
&=& \ab{\lambda x'.\ab{\lambda y'. C\SB{\Gamma,x,y\nd x}y',x'},env}\\
&=& \ab{\lambda x'.\ab{\lambda y'. C\SB{\Gamma,x\nd x}(\pi_2 (y')),x'},env}\\
&=& \ab{\lambda x'.\ab{\lambda y'. \pi_1(\pi_2 (y')),x'},env}
\end{array}\]
\end{exa}

Closure conversion seems like a natural candidate for encoding in a
logical framework, because it seems to involve only syntactic
manipulation of ordinary $\lambda$-terms.  For example,
Hannan~\cite{hannan95tpa} studied closure conversion algorithms
encoded in LF.  However, there are some subtle issues which seem to
complicate formalizing closure conversion in LF.  First, if we take
$lam : (exp \to exp) \to exp$, there is no explicit case for
variables.  This can be fixed by making sure to add a local hypothesis
$is\_var(x)$ for each $\lambda$-term variable $x$ as $x$ is added to
the context.  This approach is commonly taken in LF
developments~\cite{crary08lfmtp}, and is believed correct as long as
there is no way to construct a term of type $is\_var(M)$ where $M$ is
not a variable.  Alternatively, we could adopt a weaker encoding in
which $lam : (var \to exp) \to exp$, thus foregoing the benefits of
built-in capture-avoiding substitution.

Second, however, in LF we cannot directly test variables for equality.
Hannan~\cite{hannan95tpa} neither presented a concrete LF encoding nor
discussed how to overcome these obstacles.  Using Crary's
technique~\cite{crary08lfmtp}, we can test inequality among variables
by tagging variables with distinct numerical tags, but this requires
modifying all predicates in which inequality testing might be needed
(see the discussion in the next section).

In \DNTT, we can define closure conversion directly as a relation, as
shown in \refFig{closure-conversion}.  We use a definable type of
lists of identifiers $list\_id$, and define syntax for pairing,
projection, and $\mathbf{let}$.  The variable inequality
side-condition on the case for different variables $x,y$ is handled
using $neq$.  The rest of the translation is straightforward.

\subsection{Dynamic logic}

\begin{figure*}[tb]
\[\small\begin{array}{lcl}
  pf      &:& list\_o \to o \to \type.\\
  assignI &:& pf\ (G\conc \Ax)\ (box\ (\Ax := T\conc \Ax)\ (P\conc
  \Ax))\gets  (\new \Ay{:}v.pf\ [G\conc \Ax,var\ \Ay = T\conc \Ax]\ (P\conc \Ay)).\\
  assignE &:& pf\ (G\conc\Ax)\ (Q\conc\Ax)
\\ &\gets&
  pf\ (G\conc\Ax) \ (box\ (\Ax := T\conc \Ax)\ (P\conc \Ax))\gets(\new \Ay{:}v.pf\ [G,var\ \Ay = T\conc \Ax,P\conc\Ay]\
  (Q\conc \Ax)) .
\end{array}\]
\caption{Representative inference rules of dynamic logic}\labelFig{dynamic-logic}
\end{figure*}
 
Dynamic logic (DL)~\cite{harel00} is a generalization of program
logics such as Hoare logic.  In DL, besides ordinary propositional
connectives and quantifiers, there is a syntactic class of
\emph{programs} $\alpha$, and a modal connective $[\alpha]\phi$.  Such
a formula has the intended interpretation, ``After any terminating
execution of program $\alpha$, $\phi$
necessarily holds''.  Programs can in general be
nondeterministic or nonterminating, so $[\alpha]\phi$ is trivially
true if $\alpha$ diverges; on the other hand, $[\alpha]\phi$ does not
hold if there is a possible terminating execution of $\alpha$ in a
state not satisfying $\phi$.  Thus, a DL formula $\phi \impp
[\alpha]\psi$ has the same meaning as a Hoare logic partial
correctness assertion $\{\phi\}\alpha\{\psi\}$.

An important, but counterintuitive, aspect of dynamic logic is that
variables are used both for quantification and as assignment targets
in programs.  As a result, it does not make sense to substitute an
expression for a variable name $x$ everywhere in its scope, because it
might occur on the left-hand side of an assignment, and it would not
make sense to substitute an expression there.  For example, $\forall x.
[x := 0] (x = 0)$ is a well-formed (and valid) formula of DL, but $[1
:= 0]1 = 0$, the result of substituting a non-variable such as $1$ for
$x$, is nonsense.

Proof rules for the assignment operation $x := t$ are challenging to
encode in a logical framework.  Honsell and
Miculan~\cite{honsell96types} considered a natural deduction
formulation of DL implemented in Coq.  Their proof system included the
following inference rules to deal with assignment:
\[\small
\begin{array}{c}
\infer[{:=}I]{\Gamma \nd [x:=t]\phi}{\Gamma,y=t \nd \phi[y/x] & (y \not\in FV(\Gamma,\phi,t))}\smallskip\\
\infer[{:=}E]{\Gamma \nd \psi}{\Gamma \nd [x:=t]\phi & \Gamma,y=t,\phi[y/x] \nd \psi  & (y \not\in FV(\Gamma,\phi,\psi,t))}
\end{array}\]
The main obstacle to encoding dynamic logic using higher-order
abstract syntax is that there is no easy way to talk about distinct or
fresh object variable names.  
To deal with the freshness side conditions, Honsell and Miculan
adapted a technique introduced for encoding Hoare logic in LF by Avron, Honsell, Mason, and Pollack~\cite{mason87lfcs,avron92jar}. In this technique,
explicit judgments $isin : \Pi T{:}\type.v \to T \to \type$ and
$isnotin : \Pi T{:}\type . v \to T \to \type$ are introduced to encode the
property that a variable name occurs free in (does not occur free in)
an object of type $T$ (an expression, formula, program, etc.).  
Both LF and Coq encodings are verbose and require explicit
low-level reasoning about name occurrences, freshness, and inequality.

In \DNTT, using names and dependent name types, we can encode the
problematic inference rules as shown in \refFig{dynamic-logic}.
Again, we use a definable type of lists of formulas $list\_o$ for the
hypotheses $\Gamma$.

Here, we have taken an approach that represents the context explicitly
as part of the judgment, that is, $pf : list\_o \to o \to \type$. An alternative approach to encoding hypothetical judgments, usually
preferred in LF, is to encode only the conclusion via a predicate $pf :
o \to \type$ and then use local $pf$ assumptions to represent local
hypotheses.  
\begin{eqnarray*}
assignI' &:& pf\ (box\ (\Ax\ :=\ T\conc \Ax)\ (P\conc \Ax))
\\
&\gets&(\new \Ay{:}v.\ \ (pf\ (var\ \Ay\ =\ T\conc \Ax))\ \to\ pf\ (P\conc \Ay)).\\
\ assignE'\ &:& pf\ (Q\conc \Ax)\\
&\gets& (pf\ (box\ (\Ax\ :=\ T\conc
\Ax)\ (P\conc \Ax))) \\
&\gets& (\new \Ay{:}v.(pf\ (var\ \Ay\ =\ T\conc \Ax))\ \to\ pf\ (P\conc \Ay)\ \to\ pf\ (Q\conc \Ax)).
\end{eqnarray*}
This appears correct for \DNTT as presented in this article.  However,
if we read these types as nominal logic formulas then their meaning
does not correspond to the judgments we want to encode.  The reason is
that nominal logic satisfies an \emph{equivariance} property, which is
not explicitly reflected in \DNTT.  Equivariance states that the
validity of any proposition is preserved by applying a
name-permutation to all of its arguments.  In a type theory, this can
be represented by introducing a swapping term $\pi \act M$ such that
(roughly speaking) if $\Gamma \vdash M : A$ then $\Gamma \vdash \pi
\act M :\pi \act A$.  (This is done, in a simple type theory, for
Pitts' Nominal System T~\cite{pitts10popl,pitts11jfp}, discussed in the next
section.)  Representing hypothetical judgments using local
implications is incorrect in full nominal logic because equivariance
can be used to break the connection between names in $\Gamma$ and
names in the conclusion; to avoid this, local assumptions have to be made
explicit as an argument of the judgment.
Because we view adding additional features of nominal logic (such as
swapping/equivariance) to \DNTT as an important next step, we prefer to
give an example that appears robust in the face of these extensions.
In addition, using this approach we cannot hope to use nominal
recursion or induction principles over proofs, because of the negative
occurrences of $pf$.

Another alternative would be to represent hypotheses using
$\new$-quantification or name-abstraction:
\begin{eqnarray*}
 assignI'' &:& (pf\ (box\ (\Ax\ :=\ T\conc \Ax) (P\conc \Ax))) \\
&\gets& (\new\Ay{:}v.\  \Abs{pf\ (var\ \Ay\ =\ T\conc \Ax)}{pf\ (P\conc \Ay)}).\\
 assignE'' &:& pf\ (Q\conc \Ax)\\
& \gets& (pf\ (box\ (\Ax\ :=\ T\conc \Ax)\
 (P\conc \Ax)))\\
&\gets& ( \new\Ay{:}v.\ \Abs{pf\ (var\ \Ay\ =\
 T\conc \Ax)}{\Abs{pf\ (P\conc \Ay)}{pf\ (Q\conc \Ax)}}).
\end{eqnarray*}
Doing this would avoid the non-positivity issue, but would still have
the other drawbacks of the ordinary local hypotheses approach
discussed above.  It would also require allowing name types to depend
on values (including other names); we could do this by making
$\nametype$ into a first-class kind. However, this poses both
conceptual and practical problems.  The conceptual problem is that
name-types are usually interpreted as infinite sets of swappable
atoms, which are not mixed with ordinary values.  At a semantic level,
it is not clear what we mean by abstracting by an ordinary data type
or judgment (however, Sch\"opp's study~\cite{schoepp06lfmtp} of
nominal set semantics for Miller and Tiu's logic of generic
judgments~\cite{miller05tocl} may offer a solution).  The practical
problem is that if name-types can depend on other names, then the
context restriction operation $\wfres{\Gamma}{a}{\alpha}{\Gamma'}$
needs to remove not only all variables introduced after $\Aa$, but
also all variables or names whose type depends on $\Aa$.  This seems
workable, but makes the system considerably more complex, while it is
not yet clear that the extra complexity is justified by applications.
We view extending name-types to a first-class kind to be an important
area for future work.


\section{Comparison with related systems}\labelSec{discussion}

\subsection{LF}

We argued earlier that the intuitive definition of alpha-inequality
cannot be translated directly to LF.  This is a somewhat subjective
claim.  At a technical level, the issue is that in LF, object-language
variables are represented as meta-language variables, which cannot be
compared directly for (in)equality.  That is, we cannot simply
translate the rule
\[\infer{\var(x) \neq_\alpha \var(y)}{x \neq y}\]
directly to LF in a compositional way.  A naive attempt to represent
this rule by declaring a type constant such as
\[a  : \Pi x{:}\alpha.\Pi y{:}\alpha. \neqq~x~y.\]
is clearly wrong since this defines the total relation on expressions.
The following proposition shows that there is no way to translate
name-inequality to a binary predicate in LF that works correctly in
all contexts:
\begin{prop}
  Let $\Sigma$ be an LF signature, $t : \type$ a constant type in
  $\Sigma$ and $r: t \to t\to \type$ be a constant in $\Sigma$.  Then
  whenever $\Gamma,x{:}t,y{:}t,\Gamma' \vdash M : r~x~y$ is derivable for two different
  variables $x,y$, the judgment $\Gamma,x{:}t,\Gamma'[x/y] \vdash M[x/y]:r ~ x ~ x$ is also
  derivable.
\end{prop}
\begin{proof}
  Direct using substitution.
\end{proof}

This implies that if we want to define relations involving variable
inequality, we need to ensure that there are appropriate hypotheses in
$\Gamma$ that can be used to prove that variables introduced at
different binding sites are distinct.  For example, using Crary's
technique of adding natural number labels for bound names as they are
introduced in the context~\cite{crary08lfmtp}, we can implement
alpha-inequivalence as shown in \refFig{aneq-lf}.  (A similar encoding
is possible using weak higher-order abstract syntax techniques, as in
the Theory of Contexts~\cite{honsell01tcs}.)

\begin{figure}[tb]
\begin{verbatim}
  exp : type.
  lam : (exp -> exp) -> exp.
  app : exp -> exp -> exp.

  nat : type.
  z : nat.
  s : nat -> nat.
  neq : nat -> nat -> type.
  - : neq (s X) z.
  - : neq z (s _).
  - : neq (s N) (s M) <- neq N M.

  bvar : exp -> nat -> type.

  aneqi : nat -> exp -> exp -> type.
  - : aneqi N X Y <- bvar X MX <- bvar Y MY <- neq MX MY.
  - : aneqi N (app E1 E2) (app E3 E4) <- aneqi N E1 E3.
  - : aneqi N (app E1 E2) (app E3 E4) <- aneqi N E2 E4.
  - : aneqi N (lam E1) (lam E2) <- 
        ({x : exp} bvar X N -> aneqi (s N) (E1 x) (E2 x)).
  - : aneqi N X (app _ _)  <- bvar X _.
  - : aneqi N X (lam _)    <- bvar X _.
  - : aneqi N (app _ _) X  <- bvar X _.
  - : aneqi N (lam _) X    <- bvar X _.
  - : aneqi N (app _ _) (lam _).
  - : aneqi N (lam _) (app _ _).

  aneq : exp -> exp -> type.
  aneq_i : aneq E1 E2 <- aneqi z E1 E2.
\end{verbatim}
  
  \caption{Alpha-inequivalence in LF}
\label{fig:aneq-lf}
\end{figure}
Clearly it is a subjective question whether the other advantages of LF
outweigh the extra effort needed to encode judgments that do involve
name-inequality.  In this article, our goal has been to explore the
alternative offered by nominal abstract syntax in a dependently-typed
setting, not to propose a replacement for LF.

\subsection{Sch\"opp and Stark's dependent type theories}

Sch\"opp and Stark introduced dependent type theories that capture the
topos-theoretic semantics of nominal sets.  (The category of nominal
sets is isomorphic to the Schanuel topos, known from sheaf
theory~\cite{maclane92sheaves}).  In particular, they consider both
ordinary and ``fresh'' dependent product spaces, dependent sums, and a
``free from'' type of pairs $(a,M)$ where $a$ is a name fresh for $M$.
The ``fresh'' versions of these types quantify over objects whose
names are fresh for the current context; these generalize the
fresh-name quantifier $\new$.  The type theory is based on using bunched
contexts (derived from the Logic of Bunched Implications).

Sch\"opp and Stark's systems are very expressive: they can express
recursive functions over nominal abstract syntax, as well as proofs by
induction, as outlined earlier in this article.  But they also appear
quite difficult to use in an automated system.  In particular, there
are no results on strong normalization or decidability of equivalence
and typechecking for these systems, and it does not seem easy to adapt
standard results because of the use of bunched contexts. The results
in this paper can be seen as a first step in this direction, focusing
on a simple subsystem of theirs which captures at least some of the
expressiveness of nominal abstract syntax.

\subsection{Nominal System T and related systems}

Pitts' Nominal System T~\cite{pitts10popl,pitts11jfp} is a
simply-typed calculus that is also an attractive starting point for a
dependent nominal type theory.  In contrast to SNTT or \DNTT, it has
ordinary (non-bunched) contexts and also supports explicit name-swapping and
locally-scoped names.  Unfortunately, these features interact with
dependent types in complex ways, making it non-obvious how to extend
Nominal System T to a dependent type theory.  In this section, we give
an example that highlights the problem\footnote{This example was
  developed in informal discussions with Andrew Pitts and Stephanie Weirich}.
We give only the description of the problem, not a full formalization
of a putative ``Dependent Nominal System T.''

Consider a dependent version of Nominal System T with dependent pair
types $\Sigma x{:}A.B$ with the usual introduction and elimination rules:
\[\infer{\Gamma \vdash \ab{M,N} : \Sigma x{:}A.B}{\Gamma \vdash M : A
  & \Gamma \vdash N : B[M/x]} \qquad \infer{\Gamma \vdash
  \mathsf{unpack}~\ab{x,y} = M~\mathsf{in}~N : B'}{\Gamma \vdash M :
  \Sigma x{:}A.B & \Gamma,x:A,y:B \vdash N : B'}   \]
In Nominal System T, the $\nu$-binder can be pushed down through pair
constructors so it is natural to expect that $\nu a.\ab{M,N}$ and
$\ab{\nu a.M,\nu a.N}$ should be definitionally equal.  But if so,
then for subject reduction to hold, given a derivation of 
\[\infer{\Gamma \vdash \nu a.\ab{M,N} : \Sigma
  x{:}A.B}{\infer{\Gamma,a:\alpha \vdash \ab{M,N} : \Sigma
    x{:}A.B}{\Gamma,a:\alpha \vdash M : A & \Gamma,a:\alpha \vdash N :
    B[M/x]}}\]
we should also be able to derive 
\[\infer{\Gamma \vdash \ab{\nu a. M,\nu a. N} : \Sigma
  x{:}A.B}{\Gamma \vdash \nu a.M : A & \Gamma \vdash \nu a.N :
    B[\nu a.M/x]}\]
The first hypothesis follows immediately from $\Gamma,a:\alpha \vdash
M : A$, but it is not obvious how to obtain the second from
$\Gamma,a{:}\alpha \vdash N:B[M/x]$.  

This argument certainly does not show that it is impossible to extend
Nominal System T to a dependent type theory (doing so appears
straightforward if we limit ourselves to $\Pi$-types), just that to
develop further extensions we may need to be very careful about how
name-restrictions interact with dependent types.

\section{Conclusions}\labelSec{concl}

We have proposed a dependent nominal type theory, called \DNTT.  We
can represent name-inequality directly in \DNTT, but on the other hand
must be more explicit about contexts and substitution.  We also showed
that (recursion-free) \DNTT shares the good metatheoretic properties
of the LF type theory, particularly decidability of equivalence and
typechecking and existence of canonical forms.

There are several directions for future work.  The main syntactic
properties of the simply-typed fragment have already been verified
using Nominal Isabelle/HOL~\cite{cheney08lfmtp}.  We would also like
to relate our approach to other
techniques~\cite{pientka08popl,poswolsky08esop,licata08lics,westbrook08phd}
and further develop the foundations needed for incorporating nominal
reasoning into richer type theories such as CIC, particularly the
metatheory of recursion principles and locally-scoped names over nominal
abstract syntax.



\section*{Acknowledgements}
Thanks to Frank Nebel, Andrew Pitts, Aaron Stump, Stephanie Weirich,
and Edwin Westbrook for helpful discussions on this work.

  { \bibliographystyle{abbrv}
  \bibliography{nominal,fp,lambda,logic,logicprog,misc,lf,paper} }

\end{document}